\documentclass[]{article}

\usepackage{algorithm}
\usepackage{algorithmic}
\usepackage{amssymb}  
\usepackage{amsmath}  
\usepackage{amsthm}   
\usepackage{authblk}
\usepackage[margin = 1in]{geometry}
\usepackage[colorlinks=true, linkcolor=purple, citecolor=blue, urlcolor=blue]{hyperref}
\usepackage{natbib}
\usepackage{thm-restate}
\usepackage{xcolor}

\newtheorem{claim}{Claim}

\newtheorem{definition}{Definition}  
\newtheorem{invariant}{Invariant}
\newtheorem{lemma}{Lemma}
\newtheorem{observation}{Observation}
\newtheorem{theorem}{Theorem}

\newcommand{\MBB}{\mathrm{MBB}}
\newcommand{\p}{p}
\newcommand{\x}{X}

\title{Revisiting Fair and Efficient Allocations for Bivalued Goods}
\author{Hui Liu}
\author{Zhijie Zhang}
\affil{Fuzhou University, Fuzhou, China, \{230320053, zzhang\}@fzu.edu.cn}
\date{}

\begin{document}

\maketitle

\begin{abstract}
    This paper re-examines the problem of fairly and efficiently allocating indivisible goods among agents with additive bivalued valuations. Garg and Murhekar (2021) proposed a polynomial-time algorithm that purported to find an EFX and fPO allocation. However, we provide a counterexample demonstrating that their algorithm may fail to terminate. To address this issue, we propose a new polynomial-time algorithm that computes a WEFX (Weighted Envy-Free up to any good) and fPO allocation, thereby correcting the prior approach and offering a more general solution. Furthermore, we show that our algorithm can be adapted to compute a WEQX (Weighted Equitable up to any good) and fPO allocation.
\end{abstract}

\section{Introduction}

Since the seminal work of Steinhaus \cite{Steinhaus48}, fair division has been a central topic in mathematics, economics, and computer science \cite{Moulin03}. The field addresses the problem of distributing a set of resources, either divisible or indivisible, among a set of agents. This has wide-ranging applications, including inheritance distribution \cite{pratt1990fair}, course allocation \cite{BudishC07}, and air traffic management \cite{vossen2002fair}.

This paper focuses on the fair allocation of \emph{indivisible} goods. Formally, given a set $M$ of $m$ goods and a set $N$ of $n$ agents, the goal is to partition $M$ into $n$ disjoint bundles $X=(X_1,X_2,\ldots,X_n)$, where bundle $X_i$ is assigned to agent $i$. A fundamental requirement is that the allocation satisfies certain fairness criteria. The most intuitive notion is envy-freeness (EF) \cite{Foley67}, where no agent prefers another's bundle to her own. However, EF allocations are generally impossible to guarantee with indivisible goods (e.g., dividing a single item between two agents), leading to the study of relaxed fairness concepts.

%

Budish \cite{Budish11} introduced envy-freeness up to one good (EF1), which requires that for any pair of agents, the envy can be eliminated by removing at most one good from the other's bundle. EF1 allocations always exist and can be computed efficiently \cite{LiptonMMS04}. A stronger relaxation, envy-freeness up to \emph{any} good (EFX), was proposed by Caragiannis et al. \cite{CaragiannisKMP016}. It demands that envy be removable by the elimination of \emph{any} good from the envied bundle. The existence of EFX allocations in general remains a major open problem, though it has been established for several special cases, including when agents have bivalued valuations \cite{AmanatidisBFHV20,byrka2025,jin2025}.

Another critical research direction seeks allocations that are both fair and economically efficient, measured by Pareto optimality (PO) or its fractional strengthening (fPO). Caragiannis et al.~\cite{CaragiannisKMP016} showed that maximizing Nash Social Welfare (NSW) yields EF1 and PO allocations, proving their existence. However, maximizing NSW is computationally hard \cite{NguyenNRR14,Lee17}. For the stricter EFX notion, Amanatidis et al.~\cite{AmanatidisBFHV20} showed that maximizing NSW also yields EFX and PO allocations for bivalued instances, though the computational hurdle remains.

To achieve computational tractability, Fisher market models have been employed. This line of work has led to pseudo-polynomial algorithms for EF1 and fPO allocations \cite{BarmanKV18,MurhekarG21}. Specifically for bivalued valuations, polynomial-time algorithms for EF1 and fPO are known \cite{MurhekarG21}. Garg and Murhekar \cite{GargM21} claimed that a Fisher-market-based algorithm could produce an EFX and fPO allocation in polynomial time for bivalued instances. However, as we demonstrate, their algorithm may not terminate. We formalize this in Theorem \ref{thm: counterexample} (proof in Appendix \ref{sec: a counterexample}).


\begin{restatable}{theorem}{ThmCounterexample}
	\label{thm: counterexample}
	There exists a counterexample on which the algorithm of Garg and Murhekar \cite{GargM21} never halts.
\end{restatable}


To remedy this, we propose a new polynomial-time algorithm that computes a \emph{weighted} EFX (WEFX) and fPO allocation, as stated in Theorem \ref{thm: WEFX and fPO}. We note that an alternative algorithm for EFX and fPO allocations exists \cite{BuLLLT24}; our contribution generalizes this result to the weighted setting. Previously, such weighted guarantees were known only for binary valuations \cite{NeohT25}. The Fisher market is a powerful tool for fair and efficient allocation \cite{BarmanKV18, MurhekarG21, GargM21, FreemanSVX19n,WuZZ23,LinWZ25}, and our work further deepens its understanding.

Finally, we consider equitability up to any good (EQX) \cite{GourvesMT14m}. Gourvès et al.~\cite{GourvesMT14m} proved that EQX allocations always exist for additive valuations and can be found in polynomial time, a result later strengthened to EQX and PO \cite{FreemanSVX19n}. Garg and Murhekar \cite{GargM21} gave a polynomial-time algorithm for EQX and PO allocations for bivalued goods. By adapting our approach, we generalize this result to the weighted setting (WEQX and fPO), as captured in Theorem \ref{thm: WEQX and fPO}.

\begin{theorem}
	\label{thm: WEFX and fPO}
	There exists a polynomial-time algorithm that computes a WEFX and fPO allocation for bivalued instances.
\end{theorem}

\begin{theorem}
	\label{thm: WEQX and fPO}
	There exists a polynomial-time algorithm that computes a WEQX and fPO allocation for bivalued instances.
\end{theorem}

\paragraph{Paper Structure.} In Section \ref{sec: pre}, we introduce the formal definition of the problem as well as the Fisher market. In Section \ref{sec: WEFX and fPO}, we present our algorithm for finding WEFX and fPO allocations. In Section \ref{sec: WEQX and fPO}, we present our algorithm for finding WEQX and fPO allocations. We summarize the paper in Section \ref{sec: conclusion} and propose some open problems.

%
%
%

\section{Preliminaries}
\label{sec: pre}

Fair allocation of goods aims to fairly distribute a set $M$ of $m$ indivisible goods among a set $N$ of $n$ agents. Each agent $i \in N$ has a non-negative \emph{utility} function $v_i: 2^M \to \mathbb{R}_{\geq 0}$ that describes the agent's preferences over the goods. We assume that the utility functions are \emph{additive}, i.e.~for any agent $i\in N$ and \emph{bundle} $X\subseteq M$, $v_i(X) = \sum_{e \in X} v_i(e)$, where $v_i(e)$ is a shorthand for $v_i(\{e\})$. We further assume that the utility functions are \emph{bivalued}, i.e.~there exist two numbers $a>b\geq 0$ such that for any agent $i\in N$ and good $e\in M$, $v_i(e)\in \{a,b\}$.
In this paper, we focus on the case where $a>b>0$ and by rescaling the ulitilies, we assume that $v_i(e)\in\{1,k\}$ for some $k>1$.

An \emph{allocation} $X=(X_1,X_2,\ldots,X_n)$ is a partition of $M$ into $n$ bundles such that $X_i\cap X_j=\emptyset$ for any $i\neq j$, $\cup_{i\in N} X_i=M$, and $X_i$ is allocated to agent $i$. Assume that each agent $i\in N$ is associated with a weight $w_i>0$ satisfying $\sum_{i\in N} w_i=1$. We focus on the two well-studied concepts below for measuring the fairness of an allocation.
\begin{definition}[WEFX]
	An allocation $X$ is \emph{weighted envy-free up to any good} (WEFX) if for all $i, j \in N$ and every $e\in X_j$,
	\[ \frac{v_i(X_i)}{w_i}\geq \frac{v_i(X_j \setminus \{e\})}{w_j}. \]
\end{definition}
\begin{definition}[WEQX]
	An allocation $X$ is \emph{weighted equitable up to any good} (WEQX) if for all $i, j \in N$ and every $e\in X_j$,
	\[ \frac{v_i(X_i)}{w_i}\geq \frac{v_j(X_j \setminus \{e\})}{w_j}. \]
\end{definition}
Note that when $w_{i} =1/n$ for all $i\in N$, both concepts reduce to their unweighted versions.

Next, we introduce the well-known (fractional) Pareto optimality for measuring the efficiency of an allocation.
\begin{definition}[PO and fPO]
	We say allocation $Y$ \emph{dominates} allocation $X$ if $v_i(X_i) \leq v_i(Y_i)$ for all $i \in N$, and there is some $j \in N$ such that $v_j(X_j) < v_j(Y_j)$.
	An allocation $X$ is \emph{Pareto optimal} (PO) if no allocation dominates $X$. Furthermore, an allocation $X$ is \emph{fractionally Pareto optimal} (fPO) if no fractional allocation dominates $X$. An fPO allocation is PO, but not vice versa.
\end{definition}
\paragraph{Fisher Market.} In the Fisher market, each good $e\in M$ is assigned a price $p(e)>0$.
Given the price vector $p$, we define the \emph{bang-per-buck} ratio $\alpha_{i,e}$ of agent $i$ for good $e$ as $\alpha_{i,e} = \frac{v_i(e)}{p(e)}$, and the \emph{maximum bang-per-buck} (MBB) ratio $\alpha_i$ of agent $i$ as $\alpha_i = \max_{e \in M} \alpha_{i,e}$.
For each agent $i\in N$, we define $\MBB_i = \{ e \in M \mid \alpha_{i,e} = \alpha_i \}$ as the set of goods that achieve the MBB ratio of $i$.
$\MBB_i$ is known as the \emph{MBB set} of $i$.
We say an allocation $X$ with price vector $p$ is MBB, or equivalently, $(X,p)$ forms a \emph{Fisher market equilibrium}, if for all $i \in N$, $X_i\subseteq \MBB_i$. That is, agents are only assigned goods that are MBB for them.
The celebrated First Welfare Theorem \cite{1995Microeconomic} 
states that for any market equilibrium $(X, p)$, the allocation $X$ is fPO.
We formulate this result as the following lemma.
\begin{lemma}\label{lem: MBB_to_fPO}
	For any equilibrium $(X, p)$, $X$ is fPO.
\end{lemma}

We next define the fairness notion w.r.t.~the price and show that it is a sufficient condition for the WEFX property.

\begin{definition}[pWEFX]\label{def: pWEFX}
	An equilibrium $(X, p)$ is called \emph{price weighted envy-free up to any good (pWEFX)} if for all $i, j \in N$ and every $e \in X_j$,
	\[
	\frac{p(X_i)}{w_i} \geq \frac{p(X_j\setminus \{e\})}{w_j}.
	\]
\end{definition}
Throughout the paper, we call $p(X_{i}) / w_{i}$ the \emph{(weighted) spending} of agent $i$.
We also use $\hat{p}_{i}$ to denote the (weighted) spending of agent $i$ after removing the good with minimum price, i.e.,
\[
\hat{p}_{i} = \max_{e \in X_{i}} \left\{ \frac{p \left( X_{i} \setminus \{e\} \right)}{w_{i}} \right\}.
\]

\begin{definition}[Least and Big Spenders]\label{def: spenders}
	Given an equilibrium $(X, p)$, an agent $\ell \in N$ is called a \emph{least spender} if $\ell = \arg \min_{i \in N} \left\{ \dfrac{p(X_i)}{w_i} \right\}$; an agent $b \in N$ is called a \emph{big spender} if $b = \arg \max_{i \in N} \{\hat{p}_i\}$.
\end{definition}

\begin{lemma}
	\label{lem: pWEFX implies WEFX}
	In an equilibrium $(X,p)$, if agent $i$ is pWEFX, then $i$ is WEFX.
\end{lemma}

\begin{proof}
	Consider any $i,j\in N$.
	Let $e'\in X_j$ be such that $\hat{\p}_j=\frac{\p(\x_j\setminus\{e'\})}{w_j}$.
	Then, 
	\[ \frac{v_i(X_i)}{w_i}=\alpha_i\cdot \frac{p(X_i)}{w_i}\geq \alpha_i\cdot \frac{\p(\x_j\setminus\{e'\})}{w_j}\geq \sum_{e\in\x_j\setminus\{e'\}}\alpha_{i,e}\cdot \frac{p(e)}{w_j}=\frac{v_i(X_j\setminus\{e'\})}{w_j}. \]
	The first equality and the last inequality hold since $(X,p)$ is an equilibrium.
	The first inequality holds since $i$ is pWEFX.
\end{proof}

\section{WEFX and fPO Allocations}
\label{sec: WEFX and fPO}

In this section, we provide an algorithm that computes a WEFX and fPO allocation of bivalued goods in polynomial time. Our algorithm consists of two components (Algorithm \ref{alg: initial allocation} and Algorithm \ref{alg: WEFX and fPO}). We describe them in the two subsections below.

\subsection{Initial Equilibrium and Agent Groups}

In this section, we present Algorithm \ref{alg: initial allocation}, which computes an initial equilibrium $(X,p)$ and divides agents into groups $\{N_r\}_{r\in [R]}$. We then show that the output enjoys some good properties, and Algorithm \ref{alg: initial allocation} terminates in polynomial time.

We start with introducing some concepts.

\begin{definition}[Item Groups]\label{def: item groups}
	We say good $e \in M$ is \emph{consistently small} if for all $i \in N$, $v_i(e) = 1$. Define $M^-$ as the set of all consistently small items and $M^+=M\setminus M^-$:
	\begin{align*}
		M^- &= \{e \in M : \forall i\, \in N, v_i(e) = 1\}, \\
		M^+ &= \{e \in M : \exists i\, \in N, v_i(e) = k\}.
	\end{align*}
\end{definition}
\begin{definition}[MBB Graph]
	\label{def: MBB graph}
	The MBB graph $G_X=(N,E)$ associated with an equilibrium $(X, p)$ is a directed graph whose vertex set is $N$. There is a directed edge from $i$ to $j$ in $G_X$ if and only if $\MBB_i\cap X_j\neq\emptyset$. The edges an paths in $G_X$ are called MBB edges and MBB paths, respectively.
\end{definition}

Algorithm \ref{alg: initial allocation} starts with an integral allocation $(\x,\p)$ that maximizes the social welfare. Such an allocation can be constructed in the following way: for any $e\in M^+$, allocate it to any $i\in N$ with $v_i(e)=k$; for any $e\in M^-$, allocate it arbitrarily. We then set $p(e)=v_i(e)$ for any $e\in\x_i$.
As a result, we have $p(e)=1$ for all $e\in M^-$ and $p(e)=k$ for all $e\in M^+$.
It then constructs an MBB graph $G_{\x}=(N,E)$ corresponding to the current equilibrium $(\x,\p)$. Whenever there is a path $i_0\to i_1\to\cdots\to i_s$ in $G_{\x}$ and good $e\in \x_{i_s}$ such that $\p(\x_{i_s}\setminus\{e\})>\p(\x_{i_0})$, the algorithm transfers goods backward along the path. The tie-breaking rules in lines \ref{line: multiple paths} and \ref{line: minimum price good} ensure that the \textbf{while} loop terminates in polynomial time, yielding the initial equilibrium $(\x,\p)$.

Next, Algorithm \ref{alg: initial allocation} divides agents into groups according to $(\x,\p)$ in the following way. Let $\ell_1$ be the least spender among $N$. Let $N_1$ consists of the agents that can be reached from $\ell_1$. Next, let $\ell_2$ be the least spender among $N\setminus N_1$. Define $N_2$ in a similar way. Repeating this procedure, the algorithm will finally construct a set of agent groups $\{N_r\}_{r\in[R]}$.
We call a group with a small index a \emph{low} group, and a group with a large index a \emph{high} group. Besides, we call the agent $\ell_r$ defined during the construction of $N_r$ a \emph{representative} agent of $N_r$. We say a group is pWEFX if all agents in this group are pWEFX toward each other.

\begin{algorithm}[tb]
	\caption{Computation of Initial Equilibrium and Agent Groups}
	\label{alg: initial allocation}
	\begin{algorithmic}[1]
		\REQUIRE A bivalued instance $(N, M, \{w_i\}_{i \in N}, \{v_i\}_{i \in N})$.
		\STATE $(\x,\p)\gets$ welfare-maximizing allocation, where $p(e)=v_i(e)$ for $e\in \x_i$.
		\STATE Construct the MBB graph $G_{\x}=(N,E)$ by adding an MBB edge from $i$ to $j$ if $\MBB_i\cap \x_j\neq \emptyset$.\label{line: MBB graph}
		\WHILE{there are a path $i_0\to i_1\to \cdots \to i_s$ in $G_{\x}$ and good $e\in \x_{i_s}$ s.t.~$\frac{\p(\x_{i_s}\setminus\{e\})}{w_{i_s}}>\frac{\p(\x_{i_0})}{w_{i_0}}$}
		\STATE If there are multiple such paths, choose the one with minimum $\frac{p(\x_{i_0})}{w_{i_0}}$. \label{line: multiple paths}
		\FOR{$r=s,s-1,\ldots,1$}
		\STATE Pick a good $e\in\MBB_{i_{r-1}}\cap \x_{i_r}$, breaking ties by picking the good with minimum price. \label{line: minimum price good}
		\STATE Update $\x_{i_r}\gets \x_{i_r}\setminus\{e\}$ and $\x_{i_{r-1}}\gets \x_{i_{r-1}}\cup\{e\}$.
		\ENDFOR
		\ENDWHILE
		\STATE Initialize $R\gets 0$, $N'\gets N$.
		\WHILE{$N'\neq\emptyset$}
		\STATE Let $\ell\leftarrow\arg\min_{i \in N'} \frac{\p(\x_i)}{w_i}$, breaking ties by picking the agent with smallest index.
		\STATE Update $ R \leftarrow R + 1 $ and let $ N_R \leftarrow \{i \in N' \mid\text{there is a path in } G_X \text{ from } \ell \text{ to } i\}$.
		\STATE Update $ N' \leftarrow N' \setminus N_R $.
		\ENDWHILE
		\RETURN $(\x,\p,\{N_r\}_{r\in[R]})$.
	\end{algorithmic}
\end{algorithm}



\paragraph{The Properties of Algorithm \ref{alg: initial allocation}.} 
We start with the following simple yet useful observation.
\begin{observation}
	\label{obs: MBB set and prices}
	For any agent $i\in N$, if $\MBB_i$ contains some good with price $k$, it must contain every good with price $1$.
\end{observation}
\begin{proof}
	Assume $e\in \MBB_i$ with $p(e)=k$ and $e'\notin \MBB_i$ with $p(e')=1$. Then, $\alpha_{i,e}>\alpha_{i,e'}$, implying
	\[ v_{i}(e)> \frac{p(e)}{p(e')}\cdot v_{i}(e')=k\cdot v_{i}(e'), \]
	which is impossible since $v_i(e),v_i(e')\in\{1,k\}$.
\end{proof}
We next present a condition which guarantees the pWEFX property.
\begin{lemma}
	\label{lem: i has a path to j}
	Let $\x$ be the allocation returned by Algorithm \ref{alg: initial allocation}. For any agents $i,j\in N$, if there is a path in $G_{\x}$ from $i$ to $j$, then $i$ is pWEFX toward $j$.
\end{lemma}
\begin{proof}
	Assume that the path is $i\to i_1\to\cdots\to i_{s-1}\to j$ and $e\in\,\MBB_{i_{s-1}}\cap\,\x_j$. Then, $\frac{\p(\x_j\setminus\{e\})}{w_j}\leq \frac{\p(\x_i)}{w_i}$ since otherwise the \textbf{while} loop would not have broken. If $p(e)=1$, for any $e'\in \x_j$, 
	\begin{equation}
		\label{eq: price 1}
		\frac{\p(\x_i)}{w_i}\geq \frac{\p(\x_j)-1}{w_j}\geq\frac{\p(\x_j\setminus\{e'\})}{w_j}.
	\end{equation}
	If $p(e)=k$, assume that $p(e')=k$ for any $e'\in\x_j$, then
	\[ \frac{\p(\x_i)}{w_i}\geq \frac{\p(\x_j)-k}{w_j}=\frac{\p(\x_j\setminus\{e'\})}{w_j}. \]
	Assume that there is some $e''\in\x_j$ such that $p(e'')=1$. By Observation \ref{obs: MBB set and prices}, we must have $e''\in\,\MBB_{i_{s-1}}$. Then, Eq.~\eqref{eq: price 1} holds again by replacing $e$ with $e''$. In conclusion, $i$ is pWEFX toward $j$.
\end{proof}
We now describe the properties possessed by the output of Algorithm \ref{alg: initial allocation}.
\begin{lemma}\label{lem: properties of initial allocation}
	Let $(\x,\p,\{N_r\}_{r\in[R]})$ be the output of Algorithm \ref{alg: initial allocation}. The following properties hold:
	\begin{enumerate}
		\item $(\x, \p)$ is an equilibrium, and $\alpha_i = 1$ for all $i \in N$.
		
		\item For all $i, j \in N$, if $i \in N_r, j \in N_{r'}$ with $r<r'$, then $v_i(e) =1$ for any $e \in X_j$. 
		
		\item $M^- \subseteq \cup_{i \in N_1} \x_i$.
		\item For all $r \in [R]$, the agent group $N_r$ is pWEFX.
	\end{enumerate}
\end{lemma}
\begin{proof}
	\begin{enumerate}
		\item The algorithm starts with a welfare-maximizing allocation, which is MBB, and $\alpha_i=1$ for any agent $i\in N$. It then repeatedly transfers goods along MBB edges, during which both the MBB ratios and the MBB property are preserved. Therefore, $(\x,\p)$ is an equilibrium, and $\alpha_i = 1$ for all $i \in N$.
		\item Assume that there is a good $e \in X_j$ with $v_i(e) = k$.
		Then, $e\in M^+$ and therefore $p(e)=k$. We have $\alpha_{i,e}=1=\alpha_i$ and hence $e\in\MBB_i$.
		Thus, $i\to j$ is an MBB edge. Consider the representative agent $\ell_r$ of group $N_r$. It follows that there is a path from $\ell_r$ to $j$ that passes through $i$. However, by the construction of agent groups, $j$ should be in $N_r$ instead of $N_{r'}$. A contradiction.
		\item By a similar argument, assume that there is good $e\in M^-\cap X_j$ for some $j\in N_r$ with $r>1$. Consider the representative agent $\ell_1$  of group $N_1$. We have $\alpha_{\ell_1,e}=1=\alpha_{\ell_1}$ and hence $e\in\MBB_{\ell_1}$. Thus, $\ell_1\to j$ is an MBB edge and hence $j$ should be in $N_1$ instead of $N_{r}$. A contradiction.
		\item For any agents $i,j\in N_r$, we need to show that $\frac{\p(\x_i)}{w_i}\geq \hat{\p}_j$. Consider the representative agent $\ell_r$ of group $N_r$.
		By definition, $\ell_r$ is the least spender in $N_r$ and has a path to $j$. Thus, we have $\frac{\p(\x_i)}{w_i}\geq \frac{\p(\x_{\ell})}{w_{\ell}}$ and $\frac{\p(\x_{\ell})}{w_{\ell}}\geq \hat{\p}_j$ by Lemma \ref{lem: i has a path to j}, which imply that $i$ is pWEFX toward $j$.
	\end{enumerate}
\end{proof}
\paragraph{The Running Time of Algorithm \ref{alg: initial allocation}.}
We now show that Algorithm \ref{alg: initial allocation} has a polynomial running time.
\begin{lemma}
	Algorithm \ref{alg: initial allocation} terminates in $O(\min\{k,m\}n^2m^2)$ time.
\end{lemma}
\begin{proof}
	Observe that 1) the computation of the social welfare minimizing allocation takes $O(nm)$ time; 2) the computation of the agent groups takes $O(nm)$ time since each $N_r$ can be identified via computing a connected component of $G_X$; 3) during a \textbf{while} loop, the transfer of goods along the MBB path takes $O(m)$ time, and identifying the path as well as update the MBB graph takes $O(nm)$ time.
	Therefore, to prove the lemma, it suffices to argue that the \textbf{while} loops must break after $O(\min\{k,m\}nm)$ rounds.
	
	For convenience, we refer to an execution of the \textbf{while} loop as a round, and index the rounds by $t = 1, 2,\ldots$. 
	Given the MBB path $i_0\to i_1\to \ldots\to i_s$ in some round, we call $i_0$ the start-agent and $i_s$ the end-agent of the path. 
	Any quantity with a superscript $t$ denotes its status at the beginning of round $t$.
	\begin{claim}
		\label{claim: increasing LS}
		We have $\p(\x^t_{i_0^t})/w_{i_0^t}\leq \p(\x^{t+1}_{i_0^{t+1}})/w_{i_0^{t+1}}$.
	\end{claim}
	\begin{proof}
		Assume that $i_0^t\to i_1^t\to \cdots\to i_s^t$ is the path chosen in this round. We first show that for any $r\neq s$, $\p(\x^t_{i_r^t})/ w_{i_r^t}\leq \p(\x^{t+1}_{i_r^{t}})/w_{i_r^t}$. This claim clearly holds for $r=0$ since agent $i_0^t$ receives one more good. For $r=1,2,\ldots,s-1$, assume that the contrary holds. This happens only when $i_r^t$ receives a good $e$ with $p(e)=1$ and loses a good $e'$ with $p(e')=k$ to $i_{r-1}^t$. Then, $e'\in\MBB_{i_{r-1}^t}$, and by Observation \ref{obs: MBB set and prices}, $e\in\MBB_{i_{r-1}^t}$ as well. However, by the tie-breaking rule, the good that is transferred from $i_{r}^t$ to $i_{r-1}^t$ should be $e$ rather than $e'$. A contradiction. Next, assume that $i_s^t$ loses a good $e''$. We have $\p(\x_{i_s^t}^{t+1})/w_{i_s^t}=\p(\x_{i_s^t}^{t}\setminus\{e''\})/w_{i_s^t}>\p(\x^t_{i_0^t})/w_{i_0^t}$. The last inequality is due to the loop condition. The claim follows immediately.
	\end{proof}
	\begin{claim}
		If agent $i$ is the start-agent in both rounds $t_1$ and $t_2$, where $t_1<t_2$, then we have $\p(\x^{t_1}_{i})/w_i< \p(\x^{t_2}_{i})/w_i$.
	\end{claim}
	\begin{proof}
		Assume that $\p(\x^{t_1}_{i})/w_i= \p(\x^{t_2}_{i})/w_i$. Then, $i$ must lose some good as an end-agent in some round between $t_1$ and $t_2$. Let $t<t_2$ be the last round where $i$ loses some good $e$. Let $i_0^t$ be the corresponding start-agent at round $t$. Then, we get a contradiction due to
		\[ \p(\x^{t_2}_{i})/w_i \geq \p(\x_{i}^t\setminus\{e\})/w_i 
		>\p(\x_{i_0^t}^t)/w_{i_0^t}\geq \p(\x^{t_1}_{i})/w_i. \]
		The first inequality holds since $i$ does not lose goods after round $t$. The second holds since $i$ loses good $e$ as an end-agent. The third is due to Claim \ref{claim: increasing LS}.
	\end{proof}
	Note that for each agent $i$, the value of $\p(\x_i)/w_i$ is at most $km/w_i$. By the previous claims, each agent can be identified as a start-agent at most $km$ times because each of its appearances increase $\p(\x_i)/w_i$ by at least $1/w_i$. Hence the total number of rounds is at most $knm$. 
	
	We can also bound the number of rounds by considering the utilities of agents.
	Observe that for any $i\in N$, if $X_i$ contains $m_1$ goods with value $1$ and $m_2$ goods with value $k$, then $v_i(X_i)=m_1+km_2$. Since $m_1\geq 0, m_2\geq 0, m_1+m_2\leq m$, the number of distinct utility values that $i$ can get is at most $m^2$.
	Next, note that $\alpha_{i}=1$ for any $i\in N$ throughout the algorithm. Thus, $v_i(X_i)=p(X_i)$, and by the same argument, each agent can be identified as a start-agent at most $m^2$ times. Hence the total number of rounds is at most $nm^2$.

	In conclusion, Algorithm \ref{alg: initial allocation} terminates in $O(\min\{k,m\}n^2m^2)$ time.
\end{proof}

\subsection{The Reallocation Algorithm}

In this section, we present Algorithm \ref{alg: WEFX and fPO}, which starts from the initial equilibrium $(X^0,p^0)$ and agent groups $\{N_r\}_{r\in[R]}$ computed by Algorithm \ref{alg: initial allocation}, and constructs a WEFX equilibrium $(X,p)$ by a sequence of item reallocations and price rises.

The algorithm processes $N_r$ sequentially. When handling $N_r$, it first identifies the least spender $\ell$ in $N_r$ and the big spender $b$ among all agents.
The algorithm terminates immediately if $k\cdot \frac{p(X_{\ell})}{w_{\ell}}\geq \hat{p}_b$ i.e., $\ell$ is pWEFX toward $b$ \emph{up to a factor of $k$}.
Note that the algorithm of Garg and Murhekar \cite{GargM21} halts when $\frac{p(X_{\ell})}{w_{\ell}}\geq \hat{p}_b$, where $\ell$ denotes the least spender among all agents.
The termination condition is the key difference between our algorithm and that of Garg and Murhekar \cite{GargM21}.
Our algorithm always halts eventually since we relax the termination condition.
However, upon termination, though the pWEFX property holds for the raised groups, it does not necessarily hold for the unraised groups.
Fortunately and interestingly, we are able to show that the WEFX property is satisfied by the returned allocation. Such a phenomenon also happened in \cite{LinWZ25}, which finds approximately EFX and fPO allocations for bivalued chores.

If the termination condition is not met, i.e.~$k\cdot \frac{p(X_{\ell})}{w_{\ell}}<\hat{p}_b$, the algorithm manages to make group $N_r$ pWEFX toward $b$. It first raises prices of all goods in $\bigcup_{i \in N_r} X_i$ by a factor of $k$. Note that $\frac{p(X_{\ell})}{w_{\ell}}<\hat{p}_b$ after this step. Thus, at least $\ell$ is not pWEFX toward $b$.
The algorithm then proceeds as a \textbf{while} loop: it repeatedly transfers an item from $b$ to $\ell$ and updates their states accordingly, until $\frac{p(X_{\ell})}{w_{\ell}}\geq \hat{p}_b$ is satisfied. At this point, $N_r$ becomes pWEFX toward $b$ since $\ell$ is the least spender in $N_r$.

The good transfer in the \textbf{while} loop divides into two cases. Note that the groups higher than $N_r$ are unraised and those lower than $N_r$ are raised. If $b$ lies in an unraised group, we can pick an arbitrary good from $X_b$ and transfer it to $\ell$. If $b$ lies in a raised group, we must pick a good from $X_b\setminus X_b^0$. We show in Lemma \ref{lem: transferred good} that $X_b\setminus X^0_b\neq \emptyset$ always holds and therefore this step is valid. In both cases, the transferred good $e\in\MBB_{\ell}$ (Lemma \ref{lem: price and transferred good}) and hence $(X,p)$ always forms an equilibrium.

\begin{algorithm}[tb]
	\caption{Find a WEFX and fPO allocation}
	\label{alg: WEFX and fPO}
	\begin{algorithmic}[1]
		\REQUIRE  A bivalued instance $(N, M, \{w_i\}_{i \in N}, \{v_i\}_{i \in N})$
		\STATE Let $(X,p,\{N_r\}_{r \in [R]})$ be returned by Algorithm \ref{alg: initial allocation} 
		\STATE Initialize the set of unraised agents $U \leftarrow N$
		\FOR{$r=1,2,\ldots, R-1$}
		\STATE Let $\ell \leftarrow \arg \min_{i \in N_r} \left\{ \frac{p(X_i)}{w_i} \right\}$ be the least spender in $N_r$
		\STATE Let $b \leftarrow \arg \max_{i \in N} \{\hat{p}_i\}$ be the big spender
		\IF{$k\cdot \frac{p(X_{\ell})}{w_{\ell}}\geq \hat{p}_b$}
		\RETURN $(X,p)$
		\ENDIF
		\STATE\label{line: raise price} Raise prices of all goods in $\bigcup_{i \in N_r} X_i$ by a factor of $k$
		\STATE $U \leftarrow U \setminus N_r$
		\WHILE{$\frac{p(X_{\ell})}{w_{\ell}}<\hat{p}_b$}
		\IF{$b\in U$}
		\STATE Pick an arbitrary good $e\in X_b$
		\ELSE 
		\STATE Pick an arbitrary good $e\in X_b\setminus X_b^0$\label{line: transfer from raised agent}
		\ENDIF
		\STATE Update $X_b\gets X_b\setminus\{e\}, X_{\ell}\gets X_{\ell}\cup \{e\}$
		\STATE Let $\ell \leftarrow \arg \min_{i \in N_r} \left\{ \frac{p(X_i)}{w_i} \right\}$ be the least spender in $N_r$
		\STATE Let $b \leftarrow \arg \max_{i \in N} \{\hat{p}_i\}$ be the big spender
		\ENDWHILE
		\ENDFOR
	\end{algorithmic}
\end{algorithm}

\paragraph{The Properties of Algorithm \ref{alg: WEFX and fPO}.}

Recall that we use $b$ to denote the big spender during the execution of Algorithm \ref{alg: WEFX and fPO}.
We then consider the following definition.

\begin{definition}
	Let $Q\subseteq N$ be the set of agents that are pWEFX toward $b$, i.e.~for any $i\in Q$, $\frac{p(X_i)}{w_i}\geq \hat{\p}_b$.
\end{definition}

We first define some notations. we refer to an execution of line \ref{line: raise price} or an iteration of the \textbf{while} loop as a round, and index the rounds by $t = 0,1,2,\ldots$.
An execution of line \ref{line: raise price} is called a \emph{price-rise} round, and an iteration of the \textbf{while} loop is called a \emph{transfer} round.
We denote by $X^t$ and $p^t$ the allocation and the price at the beginning of round $t$, respectively. Denote by $\ell^t$ and $b^t$ the least spender and the big spender, respectively, identified at the beginning of round $t$.
Denote by $Q^t$ the set of agents that are pWEFX toward $b^t$.
For a price-rise round $t$, $p^{t+1}(e)=k\cdot p^t(e)$ if $e$ is raised and $p^{t+1}(e)=p^t(e)$ if unraised. Besides, $X^{t+1}=X^t$, and hence sometimes we will omit the superscript of $X$ in this round. For a transfer round $t$, some good $e^t$ is transferred from $b^t$ to $\ell^t$. Thus, $X^{t+1}_{\ell^{t}}=X^{t}_{\ell^{t}}\cup \{e^t\}, X^{t+1}_{b^{t}}=X^{t}_{b^{t}}\setminus \{e^t\}$, and $X^{t+1}_{i}=X^{t}_{i}$ for all $i\notin\{\ell^t,b^t\}$. Besides, $p^{t+1}=p^t$, and hence sometimes we will omit the superscript of $p$ in this round.
We also use $N_{\leq i}$ to denote $\cup_{j\leq i} N_j$, and define $N_{\leq i}$, $N_{\geq i}$ and $N_{>i}$ accordingly.
Recall that we denote by $\ell_{r}$ the least spender in $N_{r}$ in the initial equilibrium $(X^0,p^0)$.

We describe the properties of Algorithm \ref{alg: WEFX and fPO} by introducing a set of invariants that Algorithm \ref{alg: WEFX and fPO} maintains throughout its whole execution.

\begin{invariant}[Equilibrium Invariant]
	\label{inv: equilibrium}
	$(\x,\p)$ is an equilibrium.
\end{invariant}

\begin{invariant}[Group pWEFX Invariant]
	\label{inv: pWEFX}
	For all $r \in [R]$, agent group $N_r$ is pWEFX in $(X, p)$.
\end{invariant}

\begin{invariant}[Raised Group Invariant]
	\label{inv: raised group}
	There exists $r^*\in [R]$ such that all groups $N_1,\ldots,N_{r^*-1}$ are raised exactly once, and all groups $N_{r^*},\ldots,N_R$ are not raised. Moreover, the following properties holds:
	\begin{enumerate}
		\item Goods in $\cup_{j\in N_{<r^*}} X^0_j$ are raised exactly once, and goods in $\cup_{j\in N_{\geq r^*}} X^0_j$ are unraised.
		\item For all $i\in N$, $\alpha_{i,e}\leq 1/k$ for any $e\in \cup_{j\in N_{<r^*}} X_j^0$.
		\item For all $i\in N_{<r^*}$, $\alpha_{i,e}= 1/k$ for any $e\in \cup_{j\in N_{\geq r^*}} X_j^0$.
		\item For all $i\in N_{<r^*}$, we have $\alpha_i=1/k$ and $X_i^0\subseteq X_i$, i.e., agent $i$ did not lose any good in $X_i^0$.
		\item For all $i\in N_{\geq r^*}$, we have $\alpha_i=1$ and $X_i\subseteq X_i^0$, i.e., agent $i$ did not receive any new good.
	\end{enumerate}
\end{invariant}

\begin{invariant}[Big Spending Invariant]
	\label{inv: big spending}
	The spending of the big spender across different rounds is non-increasing: $\hat{\p}^t_{b^t}\geq \hat{\p}^{t+1}_{b^{t+1}}$.
\end{invariant}

\begin{invariant}[Least Spender Invariant]
	\label{inv: least spender}
	If $N_r$ is raised in round $t$, then any $i\in N_r$ has never been identified as a big spender in rounds $1,2,\ldots, t$.
	As a result, $X^{t'}_i=X^0_i$ for all $t'\leq t$.
\end{invariant}

\begin{invariant}[$Q$ Invariant]
	\label{inv: monotone Q}
	The set $Q$ is monotonically increasing across different rounds: $Q^t\subseteq Q^{t+1}$.
\end{invariant}

%
%

We postpone the proof of the above invariants for now and assume that they hold at the beginning of some round $t$.
We show several additional properties of the algorithm, which is helpful for showing the invariants are maintained at the end of round $t$ and proving the WEFX property of the returned allocation.
Note that when $t$ is a price-rise round, we assume that $k\cdot\frac{\p^t(X^t_{\ell^t})}{w_{\ell^t}}<\hat{p}^t_{b^t}$, since otherwise the algorithm would have stopped already.
When $t$ is a transfer round, we assume that $\frac{\p^t(X^t_{\ell^t})}{w_{\ell^t}}<\hat{p}^t_{b^t}$, since otherwise the \textbf{while} loop would have broken.
For simplicity, we sometimes omit the superscript $t$ when the context is clear.
We start by the following simple observation.

\begin{observation}
	\label{obs: l and b are not in the same group}
	Assume $\ell \in N_r$ and $b\in N_{r'}$, then $r\neq r'$.
\end{observation}

\begin{proof}
	If $r=r'$, then $\frac{p(X_{\ell})}{w_{\ell}}\geq \hat{\p}_b$ by Invariant \ref{inv: pWEFX}, which contradicts to the conditions of the \textbf{while} loop or the \textbf{if} statement.
\end{proof}

The next two lemmas are concerned with agents who become pWEFX. Specifically, Lemma \ref{lem: some unraised groups are in Q} shows that an unraised group will become pWEFX when the big spender appears in it or in a lower unraised group. Lemma \ref{lem: raised groups are in Q} shows that all raised groups become pWEFX.

\begin{lemma}
	\label{lem: some unraised groups are in Q}
	If $b\in U$ becomes the big spender at the beginning of some transfer round $t$ and assume $b\in N_s$, then $N_{\geq s}\subseteq Q^{t}$.
\end{lemma}

\begin{proof}
	For any $i\in N_s$, $\frac{p(X_i)}{w_i}\geq \hat{\p}_b$ by Invariant \ref{inv: pWEFX}. Thus, $i\in Q^t$.
	For any $i\in N_r$ with $r>s$, if $i\in Q^{t-1}$, then $i\in Q^{t}$ by Invariant \ref{inv: monotone Q}. If $i\notin Q^{t-1}$, then $i$ has never been identified as a big spender at previous transfer rounds.
	To see this, assume that $i$ is the big spender at transfer round $t'\leq t-1$.
	Then, $i\in Q^{t'}\subseteq Q^{t-1}$ by Invariant \ref{inv: monotone Q}. A contradiction.
	Together with Invariant \ref{inv: raised group}, it implies that $X^{t}_i=X^0_i$.
	We then have $\frac{p^{t}(X^{t}_i)}{w_i}=\frac{p^{t}(X^0_i)}{w_i}\geq \frac{p^{t}(X^0_{\ell_s})}{w_{\ell_s}}\geq \frac{p^{t}(X^{t}_{\ell_s})}{w_{\ell_s}}\geq \hat{p}^{t}_{b}$.
	The first inequality holds since $\ell_s$ is the least spender in $N_s$ initially and $r>s$. The second inequality holds since $N_{s}$ is unraised and $ X^{t}_{\ell_s}\subseteq X^0_{\ell_s}$ by Invariant \ref{inv: raised group}. The last inequality is due to Invariant \ref{inv: pWEFX}. Thus, $i\in Q^t$.
	To conclude, we have $N_{\geq s}\subseteq Q^{t}$.
\end{proof}

\begin{lemma}
	\label{lem: raised groups are in Q}
	Assume that $N_{<r^*}$ is raised and $N_{\geq r^*}$ is unraised. Then $N_{<r^*}\subseteq Q$.
\end{lemma}

\begin{proof}
	When processing $N_r\, (r<r^*)$ , the \textbf{while} loop terminates with $\frac{p(X_{\ell})}{w_{\ell}}\geq \hat{p}_b$. Therefore, for any $i\in N_r$, $\frac{p(X_{i})}{w_{i}}\geq \frac{p(X_{\ell})}{w_{\ell}}\geq \hat{p}_b$. Thus, $N_r\subseteq Q$ and it persists by Invariant \ref{inv: monotone Q}.
\end{proof}

The last two lemmas are crucial in understanding the behavior of the algorithm. Specifically, Lemma \ref{lem: transferred good} shows that the good transfer in line \ref{line: transfer from raised agent} is valid. Lemma \ref{lem: price and transferred good} shows that the good transfer is always along some MBB edge.

\begin{lemma}
	\label{lem: transferred good}
	In the \textbf{while} loop, if $b\notin U$, then $X_b\setminus X^0_b\neq \emptyset$. Furthermore, assuming that $\ell\in N_r$, we have $X_b\setminus X^0_b\subseteq \cup_{j\in N_{>r}} X^0_j$, i.e.~all goods in $X_b\setminus X^0_b$ were initially held by the agents in groups $N_{>r}$.
\end{lemma}
\begin{proof}
	Assume that $b\in N_{r_1}$. Since $b\notin U$, we have $r_1< r$ by Observation \ref{obs: l and b are not in the same group}.
	Assume that $X_b\setminus X^0_b= \emptyset$, then $X_b=X^0_b$ by Invariant \ref{inv: raised group}. Then, $\hat{\p}_b= k\cdot \hat{\p}^0_b\leq k\cdot \frac{p^0(X^0_{\ell_{r_1}})}{w_{\ell_{r_1}}}\leq k\cdot \frac{p^0(X^0_{\ell})}{w_{\ell}}\leq k\cdot \frac{p^0(X_{\ell})}{w_{\ell}}=\frac{p(X_{\ell})}{w_{\ell}}$, contradicting to the loop condition.
	The equality holds since $b$ has been raised and $X_b=X^0_b$. The first inequality holds since $N_{r_1}$ is pWEFX initially by Lemma \ref{lem: properties of initial allocation}. The second inequality holds since $\ell_{r_1}$ is the least spender in $N_{r_1}$ initially and $r_1<r$. The last inequality holds since $X^0_{\ell}\subseteq X_{\ell}$ by Invariant \ref{inv: raised group}. The last equality holds since $\ell$ also has been raised.
	
	For the second part of the lemma, assume that at some transfer round $t'<t$, agent $b$ serves as a least spender and she received a good $e$ that is initially held by some agent $b'$, i.e.~$e\in X^0_{b'}$. Assume that $b'\in N_{r_2}$. It suffices to show that $r_2>r$. Assume that $r_2\leq r$, then $\ell\in N_{\geq r_2}\subseteq Q$ by Lemma \ref{lem: some unraised groups are in Q}, which implies that $\ell$ is pWEFX toward $b$, contradicting to the loop condition.
\end{proof}

\begin{lemma}
	\label{lem: price and transferred good}
	After the first price-rise round, $p(e)=\{k,k^2\}$ for all $e\in M$.
	Furthermore, assume that $e$ is transferred from $b$ to $\ell$ in the \textbf{while} loop. Then $p(e)=k$ and $e\in\MBB_{\ell}$.
\end{lemma}

\begin{proof}
	Note that in $(X^0,p^0)$, $p(e)=k$ for all $e\in M^{+}$ and $p(e)=1$ for all $e\in M^{-}$. By Lemma \ref{lem: properties of initial allocation}, $M^{-}\subseteq \cup_{i\in N_1} X_i^0$.
	Since $N_1$ is raised in the first price-rise round, we have $p(e)=\{k,k^2\}$ for all $e\in M$ after this round.
	
	If $b\in U$, by Invariant \ref{inv: raised group}, we have $e\in X_b\subseteq X_b^0$ is unraised, $\alpha_{\ell}=1/k$ and $\alpha_{\ell,e}=1/k$. Thus, $p(e)=k$ and $e\in\MBB_{\ell}$.
	If $b\notin U$, assume we are in the $r$-th iteration of the \textbf{for} loop processing group $N_r$.
	By Lemma \ref{lem: transferred good}, $e\in X_b\setminus X^0_b\subseteq \cup_{j\in N_{>r}} X^0_j$. Again by Invariant \ref{inv: raised group}, $e$ is unraised, $\alpha_{\ell}=1/k$ and $\alpha_{\ell,e}=1/k$. Thus, $p(e)=k$ and $e\in\MBB_{\ell}$.
\end{proof}

\paragraph{WEFX and fPO Properties.}
Assume that all invariants are maintained throughout the execution of Algorithm \ref{alg: WEFX and fPO}. Upon termination, the algorithm outputs a tuple $(X,p,\{N_r\}_{r\in [R]}, \ell, b)$, where $N_1,\ldots, N_{r^*-1}$ are raised, $N_{r^*},\ldots, N_R$ are unraised, $\ell=\arg \min_{i \in N_{r^*}} \left\{ \frac{p(X_i)}{w_i} \right\}$ and $b = \arg \max_{i \in N} \{\hat{p}_i\}$. We now show that $(X,p)$ is WEFX and fPO.

\begin{lemma}
	\label{lem: properties of reallocation algorithm}
	When Algorithm \ref{alg: WEFX and fPO} terminates, the following properties hold:
	\begin{enumerate}
		\item There exists $s\geq r^*$ such that $N_{<r^*}\cup N_{\geq s}\subseteq Q$.
		\item For any $i\in N\setminus Q$, $X_i=X^0_i$ and $p(e)=p^0(e)$ for any $e\in X_i$.
		\item $\ell$ is also the least spender in $N\setminus Q$.
	\end{enumerate}
	
\end{lemma}

\begin{proof}
	\begin{enumerate}
		\item We have $N_{< r^*}\subseteq Q$ by Lemma \ref{lem: raised groups are in Q}.
		Let $s$ be the minimum index of the group among $N_{\geq r^*}$ that has ever contained a big spender. Then, $N_{\geq s}\subseteq Q$ by Lemma \ref{lem: some unraised groups are in Q}. Thus, $N_{<r^*}\cup N_{\geq s}\subseteq Q$.
		\item Since $N_{< r^*}\subseteq Q$, we have $N\setminus Q$ is unraised. Thus, for any $i\in N\setminus Q$, $i$ never receives any new good. Besides, $i$ never serves as a big spender, since otherwise $i$ will belong to $Q$ by Lemma \ref{lem: some unraised groups are in Q}. A contradiction. As a result, $i$ never loses any good. The previous argument implies that $X_i=X^0_i$ and $p(e)=p^0(e)$ for any $e\in X_i$.
		\item For any $i\in N\setminus Q$, we have $\frac{p(X_i)}{w_i}=\frac{p^0(X^0_i)}{w_i}\geq \frac{p^0(X^0_{\ell})}{w_{\ell}}= \frac{p(X_{\ell})}{w_{\ell}}$.
		The two equalities are due to property 2. The inequality holds since $\ell$ is the least spender in group $N_{r^*}$, the lowest group in $N\setminus Q$.
	\end{enumerate}
\end{proof}

%
%

\begin{lemma}
	\label{lem: MBB ratios}
	For any $i\in N\setminus Q$ and $j\in Q$, $\alpha_i\geq k\cdot \alpha_{i,e}$ for any $e\in X_j$.
\end{lemma}
\begin{proof}
	For any $i\in N\setminus Q$ and $j\in Q$, assume $i\in N_r$ and $j\in N_{r'}$. By Lemma \ref{lem: raised groups are in Q}, $i$ is unraised. Then, $\alpha_i=1$ by Invariant \ref{inv: raised group}.
	If $j$ is unraised, then $r<r'$ by Lemma \ref{lem: properties of reallocation algorithm}. Thus, $\alpha_{i,e}=1/k$ for any $e\in X_j$ by Invariant \ref{inv: raised group}.
	If $j$ is raised, then for any $e\in X_j$, $\alpha_{i,e}\leq 1/k$ by Invariant \ref{inv: raised group}.
	Thus, in both cases, $\alpha_i\geq k\cdot \alpha_{i,e}$ for any $e\in X_j$.
\end{proof}

\begin{lemma}
	Algorithm \ref{alg: WEFX and fPO} outputs a WEFX and fPO allocation. 
\end{lemma}
\begin{proof}
	We first show that agents in $Q$ are pWEFX, and hence WEFX by Lemma \ref{lem: pWEFX implies WEFX}. To see this, for any $i\in Q$ and $j\in N$, we have $\frac{p(X_i)}{w_i}\geq \hat{\p}_b\geq \hat{\p}_j$.
	The inequalities holds by the definitions of $Q$ and $b$, respectively.
	
	We next show that agents in $N\setminus Q$ are WEFX. Consider any $i\in N\setminus Q$ and $j\in N$. Assume that $i\in N_r,j\in N_{r'}$.
	We prove the lemma by considering two cases based on the state of $j$.
	
	If $j\in N\setminus Q$ and $r\geq r'$, we show that $i$ is pWEFX and hence WEFX toward $j$ by Lemma \ref{lem: pWEFX implies WEFX}.
	This is because it holds that $\frac{p(X_i)}{w_i}=\frac{p^0(X^0_i)}{w_i}\geq \frac{p^0(X^0_{\ell_{r'}})}{w_{\ell_{r'}}}\geq \hat{\p}^0_{b}=\hat{\p}_b$.
	The two equalities are due to Lemma \ref{lem: properties of reallocation algorithm}. The first inequality is due to $r>r'$ and the construction of groups. The second inequality holds since $N_{r'}$ is pWEFX initially by Lemma \ref{lem: properties of initial allocation}.
	
	If $j\in Q$ or $j\in N\setminus Q$ but $r<r'$, we have $\frac{p(X_i)}{w_i}\geq \frac{p(X_{\ell})}{w_{\ell}} \geq \frac{\hat{\p}_b}{k}\geq \frac{\hat{\p}_j}{k}$.
	The first inequality holds since $\ell$ is the least spender in $N\setminus Q$ by Lemma \ref{lem: properties of reallocation algorithm}. The second is due to termination condition. The last inequality holds since $b$ is the big spender. We next claim that $\alpha_{i}\geq k\cdot \alpha_{i,e}$ for any $e\in X_b$. This claim holds by Lemma \ref{lem: MBB ratios} if $j\in Q$ and by Lemma \ref{lem: properties of reallocation algorithm} and Lemma \ref{lem: properties of initial allocation} if $j\in N\setminus Q$ but $r<r'$.
	Assume that $e'\in X_j$ satisfies $\hat{\p}_j=\frac{\p(\x_j\setminus\{e'\})}{w_j}$. We have
	
	\[ \frac{v_i(\x_i)}{w_i} =\alpha_i \cdot \frac{\p(\x_i)}{w_i}\geq \frac{\alpha_i}{k}\cdot \frac{\p(\x_j\setminus\{e'\})}{w_j}
	\geq\sum_{e\in\x_j\setminus\{e'\}}\frac{v_i(e)}{p(e)}\cdot \frac{p(e)}{w_j}=\frac{v_i(\x_j\setminus\{e'\})}{w_j}. \]
	Finally, by Invariant \ref{inv: equilibrium}, $(X,p)$ is an equilibrium.
	Thus, $(X,p)$ is WEFX and fPO.

\end{proof}

\begin{lemma}
	Algorithm \ref{alg: WEFX and fPO} terminates in $O(\min\{k,m\}n^2m^2)$ time.
\end{lemma}
\begin{proof}
	First note that an iteration of the \textbf{while} loop takes $O(\log m)$ time, since it takes $O(1)$ time to transfer a good and takes $O(\log m)$ time to update the statuses of the least and big spenders.
	The \textbf{while} loop has at most $m$ iterations in total, since in each iteration, the number of goods in $N_r$ increases by $1$.
	Thus, the \textbf{while} loop takes $O(m\log m)$ times.
	Besides, it is easy to see that the \textbf{for} loop has at most $n$ iterations. Thus, the \textbf{for} loop takes $O(nm\log m)$ time.
	Since Algorithm \ref{alg: WEFX and fPO} invokes Algorithm \ref{alg: initial allocation} and the latter terminates in $O(\min\{k,m\}n^2m^2)$ time.
	Algorithm \ref{alg: WEFX and fPO} also terminates in $O(\min\{k,m\}n^2m^2)$ time.
\end{proof}

To conclude, we have the following theorem.
\begin{theorem}[Restatement of Theorem \ref{thm: WEFX and fPO}]
	Algorithm \ref{alg: WEFX and fPO} returns a WEFX and fPO allocation for bivalued goods and terminates in $O(\min\{k,m\}n^2m^2)$ time. 
\end{theorem}

\paragraph{Maintenance of Invariants.}
It is easy to see that all the invariants hold at the beginning of round $0$.
Assume that they hold at the beginning of round $t$, we now show that they are maintained at the end of round $t$.

\begin{lemma}
	The Equilibrium Invariant (Invariant \ref{inv: equilibrium}) are maintained at the end of round $t$.
\end{lemma}
\begin{proof}
	Assume we are in the $r$-th iteration of the \textbf{for} loop processing group $N_r$.
	For the price-rise round, it is easy to see that for any $j\notin N_{r}$, both $X_j$ and $\MBB_j$ remain unchanged. For any $i\in N_{r}$,
	By Invariant \ref{inv: raised group}, $X_i=X^0_i$, $\alpha_i=1/k$, and $\alpha_{i,e}=1/k$ for any $e\in X_i$.
	Thus, $X_i\subseteq \MBB_i$.
	
	
	For the transfer round, a good $e$ is transferred from $b$ to $\ell$. By Lemma \ref{lem: price and transferred good}, $e\in \MBB_{\ell}$.
	Thus, $(X,p)$ remains an equilibrium.
\end{proof}

\begin{lemma}
	The Group pWEFX Invariant (Invariant \ref{inv: pWEFX}) are maintained at the end of round $t$.
\end{lemma}
\begin{proof}
	Clearly, the price-rise rounds do not affect the invariant. It suffices to consider the transfer rounds.
	In a transfer round $t$, $X^{t+1}_{\ell^{t}}=X^{t}_{\ell^{t}}\cup \{e^t\}, X^{t+1}_{b^{t}}=X^{t}_{b^{t}}\setminus \{e^t\}$, and $X^{t+1}_{i}=X^{t}_{i}$ for all $i\notin\{\ell^t,b^t\}$.
	Assume that $\ell^t\in N_r$ and $b^t\in N_{r'}$. By Observation \ref{obs: l and b are not in the same group}, $r\neq r'$. It suffices to show that each $i\in N_r\setminus\{\ell^t\}$ remains pWEFX toward $\ell^t$, and $b^t$ remains pWEFX toward every $i\in N_{r'}\setminus\{b^t\}$.
	
	For any $i\in N_r\setminus\{\ell^t\}$, we have $\frac{p(X^{t+1}_i)}{w_i}=\frac{p(X^{t}_i)}{w_i}\geq \frac{p(X^{t}_{\ell^t})}{w_{\ell^t}}=\frac{p(X^{t+1}_{\ell^t}\setminus\{e^t\})}{w_{\ell^t}}=\hat{\p}^{t+1}_{\ell^t}$.
	The inequality holds since $\ell^t$ is the least spender in $N_r$. The last equality holds by Lemma \ref{lem: price and transferred good}.
	
	For any $i\in N_{r'}\setminus\{b^t\}$, we have $\frac{p(X^{t+1}_{b^{t}})}{w_t}=\frac{p(X^{t}_{b^{t}}\setminus\{e^t\})}{w_t}=\hat{\p}^t_{b^t}\geq \hat{\p}^{t+1}_{b^{t+1}}\geq \hat{\p}^{t+1}_{i}$.
	The second equality is due to Lemma \ref{lem: price and transferred good}. The first inequality is due to Invariant \ref{inv: big spending}.
\end{proof}

\begin{lemma}
	The Raised Group Invariant (Invariant \ref{inv: raised group}) are maintained at the end of round $t$.
\end{lemma}
\begin{proof}
	Clearly, $N_r$ is raised in line \ref{line: raise price} in the $r$-th iteration of the \textbf{for} loop.
	Thus, at the beginning of the $r^*$-th iteration of the \textbf{for} loop, $N_1,\ldots,N_{r^*-1}$ are raised exactly once, and $N_{r^*},\ldots,N_R$ are not raised.
	\begin{enumerate}
		\item For any $r<r^*$, assume that $N_r$ is raised in round $t$. By Invariant \ref{inv: least spender},
		$X^t_j=X^0_j$ for any $j\in N_r$. Thus, only goods in $\cup_{j\in N_{r}} X^0_j$ are raised in this round. The property then follows.
		\item For any $e\in \cup_{j\in N_{< r^*}} X_j^0$, note that $\alpha_{i,e}\leq 1$ in $(X^0,p^0)$ by Lemma \ref{lem: properties of initial allocation}.
		By property 1, $e$ has been raised. Thus, $\alpha_{i,e}\leq 1/k$ after $e$ is raised.
		\item For any $e\in\cup_{j\in N_{\geq r^*}} X_j^0$, $v_i(e)=1$, $p(e)=k$ and hence $\alpha_{i,e}=1/k$ in $(X^0,p^0)$ by Lemma \ref{lem: properties of initial allocation}, since $i\in N_{<r^*}$ is in a lower group than any $j\in N_{\geq r^*}$. By property 1, $e$ is not raised. Thus, $\alpha_{i,e}=1/k$ is preserved.
		\item For all $i\in N_{<r^*}$, $\alpha_i=1/k$ directly follows from properties 2 and 3.
		To show $X_i^0\subseteq X_i$, assume that $i\in N_r$ and $N_r$ is raised in round $t'$. By Invariant \ref{inv: least spender},
		$X^{t''}_i=X_i^0$ for all $t''\leq t'$.
		After round $t'$, when $i$ serves as a least spender, she receives new goods. when $i$ serves as a big spender, she loses a good $e\in X_i\setminus X^0_i$.
		In both cases, $X^0_i\subseteq X_i$ persists.
		\item For all $i\in N_{\geq r^*}$, $\alpha_i=1$ in $(X^0,p^0)$ by Lemma \ref{lem: properties of initial allocation}. Since $i$ is unraised, $i$ can only serves as a big spender and loses goods in $X_i^0$. Thus, $X_i\subseteq X^0_i$ and $\alpha_i=1$ is preserved.
	\end{enumerate}
\end{proof}

\begin{lemma}
	The Big Spending Invariant (Invariant \ref{inv: big spending}) are maintained at the end of round $t$.
\end{lemma}
\begin{proof}
	Assume we are in the $r$-th iteration of the \textbf{for} loop processing group $N_r$.
	It suffices to show that $\hat{\p}^{t}_{b^{t}}\geq \hat{\p}^{t+1}_{i}$ for any $i\in N$.
	
	For a price-rise round $t$, we have $k\cdot\frac{p^t(X^{t}_{\ell^{t}})}{w_{\ell^{t}}}<\hat{\p}^{t}_{b^{t}}$ since the \textbf{if} statement failed.
	By Invariant \ref{inv: pWEFX}, $\hat{\p}^{t}_{i}\leq \frac{\p^{t}(X^{t}_{\ell^{t}})}{w_{\ell^{t}}}$ for any $i\in N_r$.
	After the price rise, combining the previous inequalities, we have $\hat{\p}^{t+1}_{i}=k\cdot \hat{\p}^{t}_{i}\leq k\cdot\frac{p^{t}(X^{t}_{\ell^{t}})}{w_{\ell^{t}}}< \hat{\p}^{t}_{b^{t}}$ for any $i\in N_r$.
	The invariant follows by observing that $\hat{\p}^t_i$ remains unchanged for all $i\notin N_r$.
	
	For a transfer round $t$, $X^{t+1}_{\ell^{t}}=X^{t}_{\ell^{t}}\cup \{e^t\}, X^{t+1}_{b^{t}}=X^{t}_{b^{t}}\setminus \{e^t\}$, and $X^{t+1}_{i}=X^{t}_{i}$ for all $i\notin\{\ell^t,b^t\}$. Thus, only the spending of $\ell^t$ increases and it suffices to consider $\ell^t$. By Lemma \ref{lem: price and transferred good}, we have $\hat{p}^{t+1}_{\ell^{t}}=\frac{p(X^{t+1}_{\ell^{t}} \setminus \{e^t\})}{w_{\ell^{t}}}=\frac{p(X^{t}_{\ell^{t}})}{w_{\ell^{t}}} < \hat{p}^{t}_{b^t}$. The inequality is due to the loop condition.
\end{proof}

\begin{lemma}
	The Least Spender Invariant (Invariant \ref{inv: least spender}) are maintained at the end of round $t$.
\end{lemma}
\begin{proof}
	Assume that $i=b^{t'}$ at some previous round $t'\leq t$. Then, $\ell^t\in Q^{t'}\subseteq Q^t$ By Lemma \ref{lem: some unraised groups are in Q} and Invariant \ref{inv: monotone Q}. Thus, $\frac{p^t(X^t_{\ell^t})}{w_{\ell^t}}\geq \hat{\p}^t_{b^t}$, which contradicts to the conditions of the \textbf{while} loop or the \textbf{if} statement.
	As a result, $i$ has never been identified as a least spender (who receives goods) or a big spender (who loses goods) in a transfer round before. Then, $X^{t'}_i=X^0_i$ for all $t'\leq t$.
\end{proof}

\begin{lemma}
	The $Q$ Invariant (Invariant \ref{inv: monotone Q}) are maintained at the end of round $t$.
\end{lemma}
\begin{proof}
	Assume we are in the $r$-th iteration of the \textbf{for} loop processing group $N_r$.
	For a price-rise round $t$, for any $i\in Q^t$, we have $\frac{p^{t+1}(X^{t+1}_i)}{w_i}\geq \frac{p^t(X^{t}_i)}{w_i}\geq \hat{\p}^t_{b^t}\geq \hat{\p}^{t+1}_{b^{t+1}}$.
	The last inequality is due to Invariant \ref{inv: big spending}. Thus, $Q^t\subseteq Q^{t+1}$.
	
	For a transfer round $t$, the spending of agents in $Q^t\setminus\{b^t\}$ is non-decreasing and hence $Q^t\setminus\{b^t\}\subseteq Q^{t+1}$ by a similar argument.
	For $b^t$, we have $\frac{p^{t+1}(X^{t+1}_{b^t})}{w_{b^t}}= \frac{p^t(X^{t}_{b^t}\setminus\{e^t\})}{w_{b^t}}=\hat{\p}^t_{b^t}\geq \hat{\p}^{t+1}_{b^{t+1}}$.
	The second equality is due to Lemma \ref{lem: price and transferred good}. The inequality is due to Invariant \ref{inv: big spending}. Thus, $Q^t\subseteq Q^{t+1}$.
\end{proof}

\section{WEQX and fPO Allocations}
\label{sec: WEQX and fPO}

In this section, we provide an algorithm that computes a WEQX and fPO allocation of bivalued goods in polynomial time.
The algorithm also consists of two components (Algorithm \ref{alg: WEQX initial allocation} and Algorithm \ref{alg: WEQX and fPO}).
Most of their descriptions and analyses are identical as Algorithm \ref{alg: initial allocation} and Algorithm \ref{alg: WEFX and fPO}, respectively.
We describe them in the two subsections below.

\subsection{Initial Equilibrium and Agent Groups}

We first present Algorithm \ref{alg: WEQX initial allocation}, which computes an initial equilibrium $(X,p)$ and divides agents into groups $\{N_r\}_{r\in [R]}$. It is almost the same as Algorithm \ref{alg: initial allocation}, except that it compares the utilities between agents instead of their spendings.
Like Algorithm \ref{alg: initial allocation}, we show that the output of Algorithm \ref{alg: WEQX initial allocation} enjoys some good properties, and it terminates in polynomial time.
The algorithm uses the notions of item groups (Definition \ref{def: item groups}) and MBB graph (Definition  \ref{def: MBB graph}).
Denote by $\hat{v}_{i}$ the (weighted) utility of agent $i$ after removing the good with minimum value, i.e.,
\[
\hat{v}_{i} = \max_{e \in X_{i}} \left\{ \frac{v_i \left( X_{i} \setminus \{e\} \right)}{w_{i}} \right\}.
\]
The algorithm also uses the following concept.
\begin{definition}[Least and Big Valuers]\label{def: valuers}
	Given an equilibrium $(X, p)$, an agent $\ell \in N$ is called a \emph{least valuer} if $\ell = \arg \min_{i \in N} \left\{ \dfrac{v_i(X_i)}{w_i} \right\}$; an agent $b \in N$ is called a \emph{big valuer} if $b = \arg \max_{i \in N} \{\hat{v}_i\}$.
\end{definition}

\begin{algorithm}[tb]
	\caption{Computation of Initial Equilibrium and Agent Groups}
	\label{alg: WEQX initial allocation}
	\begin{algorithmic}[1]
		\REQUIRE A bivalued instance $(N, M, \{w_i\}_{i \in N}, \{v_i\}_{i \in N})$.
		\STATE $(\x,\p)\gets$ welfare-maximizing allocation, where $p(e)=v_i(e)$ for $e\in \x_i$.
		\STATE Construct the MBB graph $G_{\x}=(N,E)$ by adding an MBB edge from $i$ to $j$ if $\MBB_i\cap \x_j\neq \emptyset$.\label{line: WEQX MBB graph}
		\WHILE{there are a path $i_0\to i_1\to \cdots \to i_s$ in $G_{\x}$ and good $e\in \x_{i_s}$ s.t.~$\frac{v_{i_s}(\x_{i_s}\setminus\{e\})}{w_{i_s}}>\frac{v_{i_0}(\x_{i_0})}{w_{i_0}}$}
		\STATE If there are multiple such paths, choose the one with minimum $\frac{v_{i_0}(\x_{i_0})}{w_{i_0}}$. \label{line: WEQX multiple paths}
		\FOR{$r=s,s-1,\ldots,1$}
		\STATE Pick a good $e\in\MBB_{i_{r-1}}\cap \x_{i_r}$, breaking ties by picking the good with minimum price. \label{line: WEQX minimum price good}
		\STATE Update $\x_{i_r}\gets \x_{i_r}\setminus\{e\}$ and $\x_{i_{r-1}}\gets \x_{i_{r-1}}\cup\{e\}$.
		\ENDFOR
		\ENDWHILE
		\STATE Initialize $R\gets 0$, $N'\gets N$.
		\WHILE{$N'\neq\emptyset$}
		\STATE Let $\ell\leftarrow\arg\min_{i \in N'} \frac{v_i(\x_i)}{w_i}$, breaking ties by picking the agent with smallest index.
		\STATE Update $ R \leftarrow R + 1 $ and let $ N_R \leftarrow \{i \in N' \mid\text{there is a path in } G_X \text{ from } \ell \text{ to } i\}$.
		\STATE Update $ N' \leftarrow N' \setminus N_R $.
		\ENDWHILE
		\RETURN $(\x,\p,\{N_r\}_{r\in[R]})$.
	\end{algorithmic}
\end{algorithm}



\paragraph{The Properties of Algorithm \ref{alg: WEQX initial allocation}.} 
We first present a condition which guarantees the WEQX property.
\begin{lemma}
	\label{lem: WEQX i has a path to j}
	Let $\x$ be the allocation returned by Algorithm \ref{alg: WEQX initial allocation}. For any agents $i,j\in N$, if there is a path in $G_{\x}$ from $i$ to $j$, then $i$ is WEQX toward $j$.
\end{lemma}
\begin{proof}
	Assume that the path is $i\to i_1\to\cdots\to i_{s-1}\to j$ and $e\in\,\MBB_{i_{s-1}}\cap\,\x_j$. Then, $\frac{v_j(\x_j\setminus\{e\})}{w_j}\leq \frac{v_i(\x_i)}{w_i}$ since otherwise the \textbf{while} loop would not have broken. If $v_j(e)=1$, for any $e'\in \x_j$, 
	\begin{equation}
		\label{eq: WEQX price 1}
		\frac{v_i(\x_i)}{w_i}\geq \frac{v_j(\x_j)-1}{w_j}\geq\frac{v_j(\x_j\setminus\{e'\})}{w_j}.
	\end{equation}
	If $v_j(e)=k$, assume that $v_j(e')=k$ for any $e'\in\x_j$, then
	\[ \frac{v_i(\x_i)}{w_i}\geq \frac{v_j(\x_j)-k}{w_j}=\frac{v_j(\x_j\setminus\{e'\})}{w_j}. \]
	Assume that there is some $e''\in\x_j$ such that $v_j(e'')=1$.
	Since $e,e''\in\x_j\subseteq \MBB_j$ and $\alpha_j=1$ by Lemma \ref{lem: WEQX properties of initial allocation}, it holds that $p(e)=k$ and $p(e'')=1$.
	Since $e\in \MBB_{i_{s-1}}$, we must have $e''\in\,\MBB_{i_{s-1}}$ by Observation \ref{obs: MBB set and prices}. Then, Eq.~\eqref{eq: WEQX price 1} holds again by replacing $e$ with $e''$. In conclusion, $i$ is WEQX toward $j$.
\end{proof}
We now describe the properties possessed by the output of Algorithm \ref{alg: WEQX initial allocation}.
\begin{lemma}\label{lem: WEQX properties of initial allocation}
	Let $(\x,\p,\{N_r\}_{r\in[R]})$ be the output of Algorithm \ref{alg: WEQX initial allocation}. The following properties hold:
	\begin{enumerate}
		\item $(\x, \p)$ is an equilibrium, and $\alpha_i = 1$ for all $i \in N$.
		
		\item For all $i, j \in N$, if $i \in N_r, j \in N_{r'}$ with $r<r'$, then $v_i(e) =1$ for any $e \in X_j$. 
		
		\item $M^- \subseteq \cup_{i \in N_1} \x_i$.
		\item For all $r \in [R]$, the agent group $N_r$ is WEQX.
	\end{enumerate}
\end{lemma}
\begin{proof}
	\begin{enumerate}
		\item The algorithm starts with a welfare-maximizing allocation, which is MBB, and $\alpha_i=1$ for any agent $i\in N$. It then repeatedly transfers goods along MBB edges, during which both the MBB ratios and the MBB property are preserved. Therefore, $(\x,\p)$ is an equilibrium, and $\alpha_i = 1$ for all $i \in N$.
		\item Assume that there is a good $e \in X_j$ with $v_i(e) = k$.
		Then, $e\in M^+$ and therefore $p(e)=k$. We have $\alpha_{i,e}=1=\alpha_i$ and hence $e\in\MBB_i$.
		Thus, $i\to j$ is an MBB edge. Consider the representative agent $\ell_r$ of group $N_r$. It follows that there is a path from $\ell_r$ to $j$ that passes through $i$. However, by the construction of agent groups, $j$ should be in $N_r$ instead of $N_{r'}$. A contradiction.
		\item By a similar argument, assume that there is good $e\in M^-\cap X_j$ for some $j\in N_r$ with $r>1$. Consider the representative agent $\ell_1$  of group $N_1$. We have $\alpha_{\ell_1,e}=1=\alpha_{\ell_1}$ and hence $e\in\MBB_{\ell_1}$. Thus, $\ell_1\to j$ is an MBB edge and hence $j$ should be in $N_1$ instead of $N_{r}$. A contradiction.
		\item For any agents $i,j\in N_r$, we need to show that $\frac{v_i(\x_i)}{w_i}\geq \hat{v}_j$. Consider the representative agent $\ell_r$ of group $N_r$.
		By definition, $\ell_r$ is the least valuer in $N_r$ and has a path to $j$. Thus, we have $\frac{v_i(\x_i)}{w_i}\geq \frac{v_{\ell}(\x_{\ell})}{w_{\ell}}$ and $\frac{v_{\ell}(\x_{\ell})}{w_{\ell}}\geq \hat{v}_j$ by Lemma \ref{lem: WEQX i has a path to j}, which imply that $i$ is WEQX toward $j$.
	\end{enumerate}
\end{proof}
\paragraph{The Running Time of Algorithm \ref{alg: WEQX initial allocation}.}
We now show that Algorithm \ref{alg: WEQX initial allocation} has a polynomial running time.
\begin{lemma}
	Algorithm \ref{alg: WEQX initial allocation} terminates in $O(\min\{k,m\}n^2m^2)$ time.
\end{lemma}
\begin{proof}
	Observe that 1) the computation of the social welfare minimizing allocation takes $O(nm)$ time; 2) the computation of the agent groups takes $O(nm)$ time since each $N_r$ can be identified via computing a connected component of $G_X$; 3) during a \textbf{while} loop, the transfer of goods along the MBB path takes $O(m)$ time, and identifying the path as well as update the MBB graph takes $O(nm)$ time.
	Therefore, to prove the lemma, it suffices to argue that the \textbf{while} loops must break after $O(\min\{k,m\}nm)$ rounds.
	
	For convenience, we refer to an execution of the \textbf{while} loop as a round, and index the rounds by $t = 1, 2,\ldots$. 
	Given the MBB path $i_0\to i_1\to \ldots\to i_s$ in some round, we call $i_0$ the start-agent and $i_s$ the end-agent of the path. 
	Any quantity with a superscript $t$ denotes its status at the beginning of round $t$.
	\begin{claim}
		\label{claim: WEQX increasing LS}
		We have $v_{i_0^t}(\x^t_{i_0^t})/w_{i_0^t}\leq v_{i_0^t}(\x^{t+1}_{i_0^{t+1}})/w_{i_0^{t+1}}$.
	\end{claim}
	\begin{proof}
		Assume that $i_0^t\to i_1^t\to \cdots\to i_s^t$ is the path chosen in this round. We first show that for any $r\neq s$, $v_{i_r^t}(\x^t_{i_r^t})/ w_{i_r^t}\leq v_{i_r^t}(\x^{t+1}_{i_r^{t}})/w_{i_r^t}$. This claim clearly holds for $r=0$ since agent $i_0^t$ receives one more good. For $r=1,2,\ldots,s-1$, assume that the contrary holds. This happens only when $i_r^t$ receives a good $e$ with $v_{i_r^t}(e)=1$ and loses a good $e'$ with $v_{i_r^t}(e')=k$ to $i_{r-1}^t$.
		Since $e,e'\in\x_{i_r^t}\subseteq \MBB_j$ and $\alpha_{i_r^t}=1$ by Lemma \ref{lem: WEQX properties of initial allocation}, we have $p(e)=1$ and $p(e')=k$.
		Besides, $e'\in\MBB_{i_{r-1}^t}$ since the transfer is along an MBB edge. By Observation \ref{obs: MBB set and prices}, $e\in\MBB_{i_{r-1}^t}$ as well. However, by the tie-breaking rule, the good that is transferred from $i_{r}^t$ to $i_{r-1}^t$ should be $e$ rather than $e'$. A contradiction. Next, assume that $i_s^t$ loses a good $e''$. We have $v_{i_s^t}(\x_{i_s^t}^{t+1})/w_{i_s^t}=v_{i_s^t}(\x_{i_s^t}^{t}\setminus\{e''\})/w_{i_s^t}>v_{i_0^t}(\x^t_{i_0^t})/w_{i_0^t}$. The last inequality is due to the loop condition. The claim follows immediately.
	\end{proof}
	\begin{claim}
		If agent $i$ is the start-agent in both rounds $t_1$ and $t_2$, where $t_1<t_2$, then we have $v_{i}(\x^{t_1}_{i})/w_i< v_{i}(\x^{t_2}_{i})/w_i$.
	\end{claim}
	\begin{proof}
		Assume that $\p(\x^{t_1}_{i})/w_i= \p(\x^{t_2}_{i})/w_i$. Then, $i$ must lose some good as an end-agent in some round between $t_1$ and $t_2$. Let $t<t_2$ be the last round where $i$ loses some good $e$. Let $i_0^t$ be the corresponding start-agent at round $t$. Then, we get a contradiction due to
		\[ v_{i}(\x^{t_2}_{i})/w_i \geq v_{i}(\x_{i}^t\setminus\{e\})/w_i 
		>v_{i}(\x_{i_0^t}^t)/w_{i_0^t}\geq v_{i}(\x^{t_1}_{i})/w_i. \]
		The first inequality holds since $i$ does not lose goods after round $t$. The second holds since $i$ loses good $e$ as an end-agent. The third is due to Claim \ref{claim: WEQX increasing LS}.
	\end{proof}
	Note that for each agent $i$, the value of $v_i(\x_i)/w_i$ is at most $km/w_i$. By the previous claims, each agent can be identified as a start-agent at most $km$ times because each of its appearances increase $v_i(\x_i)/w_i$ by at least $1/w_i$. Hence the total number of rounds is at most $knm$. 
	
	On the other hand, observe that for any $i\in N$, if $X_i$ contains $m_1$ goods with value $1$ and $m_2$ goods with value $k$, then $v_i(X_i)=m_1+km_2$. Since $m_1\geq 0, m_2\geq 0, m_1+m_2\leq m$, the number of distinct utility values that $i$ can get is at most $m^2$.
	By the same argument, each agent can be identified as a start-agent at most $m^2$ times. Hence the total number of rounds is at most $nm^2$.

	In conclusion, Algorithm \ref{alg: WEQX initial allocation} terminates in $O(\min\{k,m\}n^2m^2)$ time.
\end{proof}

\subsection{The Reallocation Algorithm}

In this section, we present Algorithm \ref{alg: WEQX and fPO}, which starts from the initial equilibrium $(X^0,p^0)$ and agent groups $\{N_r\}_{r\in[R]}$ computed by Algorithm \ref{alg: WEQX initial allocation}, and constructs a WEQX equilibrium $(X,p)$ by a sequence of item reallocations and price rises.
It is almost the same as Algorithm \ref{alg: WEFX and fPO}, with two minor modifications. First, it considers the least and big valuers and compare their utilities instead of the spendings of the spenders. Second, since it directly handle the utilities, the termination condition requires that $\ell$ is WEQX toward $b$, which is not relaxed like in Algorithm \ref{alg: WEFX and fPO}. 

\begin{algorithm}[tb]
	\caption{Find a WEQX and fPO allocation}
	\label{alg: WEQX and fPO}
	\begin{algorithmic}[1]
		\REQUIRE  A bivalued instance $(N, M, \{w_i\}_{i \in N}, \{v_i\}_{i \in N})$
		\STATE Let $(X,p,\{N_r\}_{r \in [R]})$ be returned by Algorithm \ref{alg: WEQX initial allocation} 
		\STATE Initialize the set of unraised agents $U \leftarrow N$
		\FOR{$r=1,2,\ldots, R-1$}
		\STATE Let $\ell \leftarrow \arg \min_{i \in N_r} \left\{ \frac{v_i(X_i)}{w_i} \right\}$ be the least valuer in $N_r$
		\STATE Let $b \leftarrow \arg \max_{i \in N} \{\hat{v}_i\}$ be the big valuer
		\IF{$\frac{v_{\ell}(X_{\ell})}{w_{\ell}}\geq \hat{v}_b$}
		\RETURN $(X,p)$
		\ENDIF
		\STATE\label{line: WEQX raise price} Raise prices of all goods in $\bigcup_{i \in N_r} X_i$ by a factor of $k$
		\STATE $U \leftarrow U \setminus N_r$
		\WHILE{$\frac{v_{\ell}(X_{\ell})}{w_{\ell}}<\hat{v}_b$}
		\IF{$b\in U$}
		\STATE Pick an arbitrary good $e\in X_b$
		\ELSE 
		\STATE Pick an arbitrary good $e\in X_b\setminus X_b^0$\label{line: WEQX transfer from raised agent}
		\ENDIF
		\STATE Update $X_b\gets X_b\setminus\{e\}, X_{\ell}\gets X_{\ell}\cup\{e\}$
		\STATE Let $\ell \leftarrow \arg \min_{i \in N_r} \left\{ \frac{v_i(X_i)}{w_i} \right\}$ be the least valuer in $N_r$
		\STATE Let $b \leftarrow \arg \max_{i \in N} \{\hat{v}_i\}$ be the big valuer
		\ENDWHILE
		\ENDFOR
	\end{algorithmic}
\end{algorithm}

\paragraph{The Properties of Algorithm \ref{alg: WEQX and fPO}.}

Recall that we use $b$ to denote the big valuer during the execution of Algorithm \ref{alg: WEQX and fPO}.
We then consider the following definition.

\begin{definition}
	Let $Q\subseteq N$ be the set of agents that are WEQX toward $b$, i.e.~for any $i\in Q$, $\frac{v_i(X_i)}{w_i}\geq \hat{v}_b$.
\end{definition}

We first define some notations. we refer to an execution of line \ref{line: WEQX raise price} or an iteration of the \textbf{while} loop as a round, and index the rounds by $t = 0,1,2,\ldots$.
An execution of line \ref{line: WEQX raise price} is called a \emph{price-rise} round, and an iteration of the \textbf{while} loop is called a \emph{transfer} round.
We denote by $X^t$ and $p^t$ the allocation and the price at the beginning of round $t$, respectively. Denote by $\ell^t$ and $b^t$ the least valuer and the big valuer, respectively, identified at the beginning of round $t$.
Denote by $Q^t$ the set of agents that are WEQX toward $b^t$.
For a price-rise round $t$, $p^{t+1}(e)=k\cdot p^t(e)$ if $e$ is raised and $p^{t+1}(e)=p^t(e)$ if unraised. Besides, $X^{t+1}=X^t$, and hence sometimes we will omit the superscript of $X$ in this round. For a transfer round $t$, some good $e^t$ is transferred from $b^t$ to $\ell^t$. Thus, $X^{t+1}_{\ell^{t}}=X^{t}_{\ell^{t}}\cup \{e^t\}, X^{t+1}_{b^{t}}=X^{t}_{b^{t}}\setminus \{e^t\}$, and $X^{t+1}_{i}=X^{t}_{i}$ for all $i\notin\{\ell^t,b^t\}$. Besides, $p^{t+1}=p^t$, and hence sometimes we will omit the superscript of $p$ in this round.
We also use $N_{\leq i}$ to denote $\cup_{j\leq i} N_j$, and define $N_{\leq i}$, $N_{\geq i}$ and $N_{>i}$ accordingly.
Recall that we denote by $\ell_{r}$ the least valuer in $N_{r}$ in the initial equilibrium $(X^0,p^0)$.

We describe the properties of Algorithm \ref{alg: WEQX and fPO} by introducing a set of invariants that Algorithm \ref{alg: WEQX and fPO} maintains throughout its whole execution.

\begin{invariant}[Equilibrium Invariant]
	\label{inv: WEQX equilibrium}
	$(\x,\p)$ is an equilibrium.
\end{invariant}

\begin{invariant}[Group WEQX Invariant]
	\label{inv: WEQX}
	For all $r \in [R]$, agent group $N_r$ is WEQX in $(X, p)$.
\end{invariant}

\begin{invariant}[Raised Group Invariant]
	\label{inv: WEQX raised group}
	There exists $r^*\in [R]$ such that all groups $N_1,\ldots,N_{r^*-1}$ are raised exactly once, and all groups $N_{r^*},\ldots,N_R$ are not raised. Moreover, the following properties holds:
	\begin{enumerate}
		\item Goods in $\cup_{j\in N_{<r^*}} X^0_j$ are raised exactly once, and goods in $\cup_{j\in N_{\geq r^*}} X^0_j$ are unraised.
		\item For all $i\in N$, $\alpha_{i,e}\leq 1/k$ for any $e\in \cup_{j\in N_{<r^*}} X_j^0$.
		\item For all $i\in N_{<r^*}$, $\alpha_{i,e}= 1/k$ for any $e\in \cup_{j\in N_{\geq r^*}} X_j^0$.
		\item For all $i\in N_{<r^*}$, we have $\alpha_i=1/k$ and $X_i^0\subseteq X_i$, i.e., agent $i$ did not lose any good in $X_i^0$.
		\item For all $i\in N_{\geq r^*}$, we have $\alpha_i=1$ and $X_i\subseteq X_i^0$, i.e., agent $i$ did not receive any new good.
	\end{enumerate}
\end{invariant}

\begin{invariant}[Big Value Invariant]
	\label{inv: WEQX big spending}
	The value of the big valuer across different rounds is non-increasing: $\hat{v}^t_{b^t}\geq \hat{v}^{t+1}_{b^{t+1}}$.
\end{invariant}

\begin{invariant}[Least Valuer Invariant]
	\label{inv: WEQX least valuer}
	If $N_r$ is raised in round $t$, then any $i\in N_r$ has never been identified as a big valuer in rounds $1,2,\ldots, t$.
	As a result, $X^{t'}_i=X^0_i$ for all $t'\leq t$.
\end{invariant}

\begin{invariant}[$Q$ Invariant]
	\label{inv: WEQX monotone Q}
	The set $Q$ is monotonically increasing across different rounds: $Q^t\subseteq Q^{t+1}$.
\end{invariant}

%
%

We postpone the proof of the above invariants for now and assume that they hold at the beginning of some round $t$.
We show several additional properties of the algorithm, which is helpful for showing the invariants are maintained at the end of round $t$ and proving the WEQX property of the returned allocation.
Note that we can assume $\frac{v_{\ell^t}^t(X^t_{\ell^t})}{w_{\ell^t}}<\hat{v}^t_{b^t}$ at the beginning of round $t$, since otherwise the round would not been executed.
For simplicity, we sometimes omit the superscript $t$ when the context is clear.
We start by the following simple observation.

\begin{observation}
	\label{obs: WEQX l and b are not in the same group}
	Assume $\ell \in N_r$ and $b\in N_{r'}$, then $r\neq r'$.
\end{observation}

\begin{proof}
	If $r=r'$, then $\frac{v_{\ell}(X_{\ell})}{w_{\ell}}\geq \hat{v}_b$ by Invariant \ref{inv: WEQX}, which contradicts to the conditions of the \textbf{while} loop or the \textbf{if} statement.
\end{proof}

The next two lemmas are concerned with agents who become WEQX. Specifically, Lemma \ref{lem: WEQX some unraised groups are in Q} shows that an unraised group will become WEQX when the big valuer appears in it or in a lower unraised group. Lemma \ref{lem: WEQX raised groups are in Q} shows that all raised groups become WEQX.

\begin{lemma}
	\label{lem: WEQX some unraised groups are in Q}
	If $b\in U$ becomes the big valuer at the beginning of some transfer round $t$ and assume $b\in N_s$, then $N_{\geq s}\subseteq Q^{t}$.
\end{lemma}

\begin{proof}
	For any $i\in N_s$, $\frac{v_i(X_i)}{w_i}\geq \hat{v}_b$ by Invariant \ref{inv: WEQX}. Thus, $i\in Q^t$.
	For any $i\in N_r$ with $r>s$, if $i\in Q^{t-1}$, then $i\in Q^{t}$ by Invariant \ref{inv: WEQX monotone Q}. If $i\notin Q^{t-1}$, then $i$ has never been identified as a big valuer at previous transfer rounds.
	To see this, assume that $i$ is the big valuer at transfer round $t'\leq t-1$.
	Then, $i\in Q^{t'}\subseteq Q^{t-1}$ by Invariant \ref{inv: WEQX monotone Q}. A contradiction.
	Together with Invariant \ref{inv: WEQX raised group}, it implies that $X^{t}_i=X^0_i$.
	We then have $\frac{v_i(X^{t}_i)}{w_i}=\frac{v_i(X^0_i)}{w_i}\geq \frac{v_{\ell_s}(X^0_{\ell_s})}{w_{\ell_s}}\geq \frac{v_{\ell_s}(X^{t}_{\ell_s})}{w_{\ell_s}}\geq \hat{v}^{t}_{b}$.
	The first inequality holds since $\ell_s$ is the least valuer in $N_s$ initially and $r>s$. The second inequality holds since $N_{s}$ is unraised and $ X^{t}_{\ell_s}\subseteq X^0_{\ell_s}$ by Invariant \ref{inv: WEQX raised group}. The last inequality is due to Invariant \ref{inv: WEQX}. Thus, $i\in Q^t$.
	To conclude, we have $N_{\geq s}\subseteq Q^{t}$.
\end{proof}

\begin{lemma}
	\label{lem: WEQX raised groups are in Q}
	Assume that $N_{<r^*}$ is raised and $N_{\geq r^*}$ is unraised. Then $N_{<r^*}\subseteq Q$.
\end{lemma}

\begin{proof}
	When processing $N_r\, (r<r^*)$ , the \textbf{while} loop terminates with $\frac{v_{\ell}(X_{\ell})}{w_{\ell}}\geq \hat{v}_b$. Therefore, for any $i\in N_r$, $\frac{v_i(X_{i})}{w_{i}}\geq \frac{v_{\ell}(X_{\ell})}{w_{\ell}}\geq \hat{p}_b$. Thus, $N_r\subseteq Q$ and it persists by Invariant \ref{inv: WEQX monotone Q}.
\end{proof}

The last two lemmas are crucial in understanding the behavior of the algorithm. Specifically, Lemma \ref{lem: WEQX transferred good} shows that the good transfer in line \ref{line: WEQX transfer from raised agent} is valid. Lemma \ref{lem: WEQX price and transferred good} shows that the good transfer is always along some MBB edge.

\begin{lemma}
	\label{lem: WEQX transferred good}
	In the \textbf{while} loop, if $b\notin U$, then $X_b\setminus X^0_b\neq \emptyset$. Furthermore, assuming that $\ell\in N_r$, we have $X_b\setminus X^0_b\subseteq \cup_{j\in N_{>r}} X^0_j$, i.e.~all goods in $X_b\setminus X^0_b$ were initially held by the agents in groups $N_{>r}$.
\end{lemma}
\begin{proof}
	Assume that $b\in N_{r_1}$. Since $b\notin U$, we have $r_1< r$ by Observation \ref{obs: WEQX l and b are not in the same group}.
	Assume that $X_b\setminus X^0_b= \emptyset$, then $X_b=X^0_b$ by Invariant \ref{inv: WEQX raised group}. Then, $\hat{v}_b^t= \hat{v}^0_b\leq \frac{v_{\ell_{r_1}}(X^0_{\ell_{r_1}})}{w_{\ell_{r_1}}}\leq \frac{v_{\ell}(X^0_{\ell})}{w_{\ell}}\leq \frac{v_{\ell}(X_{\ell})}{w_{\ell}}$, contradicting to the loop condition.
	The equality holds since $X_b=X^0_b$. The first inequality holds since $N_{r_1}$ is WEQX initially by Lemma \ref{lem: WEQX properties of initial allocation}. The second inequality holds since $\ell_{r_1}$ is the least valuer in $N_{r_1}$ initially and $r_1<r$. The last inequality holds since $X^0_{\ell}\subseteq X_{\ell}$ by Invariant \ref{inv: WEQX raised group}.
	
	For the second part of the lemma, assume that at some transfer round $t'<t$, agent $b$ serves as a least valuer and she received a good $e$ that is initially held by some agent $b'$, i.e.~$e\in X^0_{b'}$. Assume that $b'\in N_{r_2}$. It suffices to show that $r_2>r$. Assume that $r_2\leq r$, then $\ell\in N_{\geq r_2}\subseteq Q$ by Lemma \ref{lem: WEQX some unraised groups are in Q}, which implies that $\ell$ is WEQX toward $b$, contradicting to the loop condition.
\end{proof}

\begin{lemma}
	\label{lem: WEQX price and transferred good}
	After the first price-rise round, $p(e)=\{k,k^2\}$ for all $e\in M$.
	Furthermore, assume that $e$ is transferred from $b$ to $\ell$ in the \textbf{while} loop. Then, $v_{\ell}(e)=1$, $e\in\MBB_{\ell}$ and from $b$'s viewpoint, $e$ has a minimum value among goods held by $b$.
	
\end{lemma}

\begin{proof}
	Note that in $(X^0,p^0)$, $p(e)=k$ for all $e\in M^{+}$ and $p(e)=1$ for all $e\in M^{-}$. By Lemma \ref{lem: WEQX properties of initial allocation}, $M^{-}\subseteq \cup_{i\in N_1} X_i^0$.
	Since $N_1$ is raised in the first price-rise round, we have $p(e)=\{k,k^2\}$ for all $e\in M$ after this round.
	
	If $b\in U$, by Invariant \ref{inv: WEQX raised group}, we have $e\in X_b\subseteq X_b^0$ is unraised, $\alpha_{\ell}=1/k$, $\alpha_{\ell,e}=1/k$ and $\alpha_{b}=1$. Thus, $p(e)=k$, $v_{\ell}(e)=1$ and $e\in\MBB_{\ell}$. Besides, for any $e'\in X_b$, $v_b(e')=p(e')=k$. Thus, $v_b(e)$ is the smallest.
	
	If $b\notin U$, assume we are in the $r$-th iteration of the \textbf{for} loop processing group $N_r$.
	By Lemma \ref{lem: WEQX transferred good}, $e\in X_b\setminus X^0_b\subseteq \cup_{j\in N_{>r}} X^0_j$. Again by Invariant \ref{inv: WEQX raised group}, $e$ is unraised, $\alpha_{\ell}=1/k$ and $\alpha_{\ell,e}=1/k$. Thus, $p(e)=k$, $v_{\ell}(e)=1$ and $e\in\MBB_{\ell}$.
	Besides, we have $v_b(e)=1$ since $\alpha_b=1/k$ and $p(e)=k$. This again implies that $v_b(e)$ is the smallest.
\end{proof}

\paragraph{WEQX and fPO Properties.}
Assume that all invariants are maintained throughout the execution of Algorithm \ref{alg: WEQX and fPO}. We now show that $(X,p)$ is WEQX and fPO.

\begin{lemma}
	Algorithm \ref{alg: WEQX and fPO} outputs a WEQX and fPO allocation. 
\end{lemma}

\begin{proof}
	Assume that Algorithm \ref{alg: WEQX and fPO} outputs a tuple $(X,p,\{N_r\}_{r\in [R]}, \ell, b)$, where $N_1,\ldots, N_{r^*-1}$ are raised, $N_{r^*},\ldots, N_R$ are unraised, $\ell=\arg \min_{i \in N_{r^*}} \left\{ \frac{v_i(X_i)}{w_i} \right\}$ and $b = \arg \max_{i \in N} \{\hat{v}_i\}$.
	We now show that $N_{\geq r^*}\subseteq Q$.
	Let $i\in N_r$ for some $r\geq r^*$. If some agent in $N_{\leq r}\setminus N_{<r^*}$ has been served as a big valuer, then $i\in N_r\subseteq Q$ by Lemma \ref{lem: WEQX some unraised groups are in Q}. If no agent in $N_{\leq r}\setminus N_{<r^*}$ has been served as a big valuer, then for any $j\in N_{\leq r}\setminus N_{<r^*}$, $j$ never loses goods. It follows that $X_j=X_j^0$ by Invariant \ref{inv: WEQX raised group} and $\ell$ is indeed the least valuer in $N_{\leq r}\setminus N_{<r^*}$. Thus, we have
	$\frac{v_i(X_i)}{w_i}=\frac{v_i(X^0_i)}{w_i}\geq \frac{v_{\ell}(X^0_{\ell})}{w_{\ell}}= \frac{v_{\ell}(X_{\ell})}{w_{\ell}}\geq\hat{v}_b$, which means $i\in Q$.
	The last inequality is due to the termination condition.
	On the other hand, $N_{< r^*}\subseteq Q$ by Lemma \ref{lem: WEQX raised groups are in Q}. Thus, all agents are WEQX.
	Finally, by Invariant \ref{inv: WEQX equilibrium}, $(X,p)$ is an equilibrium.
	Thus, $(X,p)$ is WEFX and fPO.
\end{proof}

\begin{lemma}
	Algorithm \ref{alg: WEQX and fPO} terminates in $O(\min\{k,m\}n^2m^2)$ time.
\end{lemma}
\begin{proof}
	First note that an iteration of the \textbf{while} loop takes $O(\log m)$ time, since it takes $O(1)$ time to transfer a good and takes $O(\log m)$ time to update the statuses of the least and big valuers.
	The \textbf{while} loop has at most $m$ iterations in total, since in each iteration, the number of goods in $N_r$ increases by $1$.
	Thus, the \textbf{while} loop takes $O(m\log m)$ times.
	Besides, it is easy to see that the \textbf{for} loop has at most $n$ iterations. Thus, the \textbf{for} loop takes $O(nm\log m)$ time.
	Since Algorithm \ref{alg: WEQX and fPO} invokes Algorithm \ref{alg: WEQX initial allocation} and the latter terminates in $O(\min\{k,m\}n^2m^2)$ time.
	Algorithm \ref{alg: WEQX and fPO} also terminates in $O(\min\{k,m\}n^2m^2)$ time.
\end{proof}

To conclude, we have the following theorem.
\begin{theorem}[Restatement of Theorem \ref{thm: WEQX and fPO}]
	Algorithm \ref{alg: WEQX and fPO} returns a WEQX and fPO allocation for bivalued goods and terminates in $O(\min\{k,m\}n^2m^2)$ time. 
\end{theorem}

\paragraph{Maintenance of Invariants.}
It is easy to see that all the invariants hold at the beginning of round $0$.
Assume that they hold at the beginning of round $t$, we now show that they are maintained at the end of round $t$.

\begin{lemma}
	The Equilibrium Invariant (Invariant \ref{inv: WEQX equilibrium}) are maintained at the end of round $t$.
\end{lemma}
\begin{proof}
	Assume we are in the $r$-th iteration of the \textbf{for} loop processing group $N_r$.
	For the price-rise round, it is easy to see that for any $j\notin N_{r}$, both $X_j$ and $\MBB_j$ remain unchanged. For any $i\in N_{r}$,
	By Invariant \ref{inv: WEQX raised group}, $X_i=X^0_i$, $\alpha_i=1/k$, and $\alpha_{i,e}=1/k$ for any $e\in X_i$.
	Thus, $X_i\subseteq \MBB_i$.
	
	
	For the transfer round, a good $e$ is transferred from $b$ to $\ell$. By Lemma \ref{lem: WEQX price and transferred good}, $e\in \MBB_{\ell}$.
	Thus, $(X,p)$ remains an equilibrium.
\end{proof}

\begin{lemma}
	The Group WEQX Invariant (Invariant \ref{inv: WEQX}) are maintained at the end of round $t$.
\end{lemma}
\begin{proof}
	Clearly, the price-rise rounds do not affect the invariant. It suffices to consider the transfer rounds.
	In a transfer round $t$, $X^{t+1}_{\ell^{t}}=X^{t}_{\ell^{t}}\cup \{e^t\}, X^{t+1}_{b^{t}}=X^{t}_{b^{t}}\setminus \{e^t\}$, and $X^{t+1}_{i}=X^{t}_{i}$ for all $i\notin\{\ell^t,b^t\}$.
	Assume that $\ell^t\in N_r$ and $b^t\in N_{r'}$. By Observation \ref{obs: WEQX l and b are not in the same group}, $r\neq r'$. It suffices to show that each $i\in N_r\setminus\{\ell^t\}$ remains WEQX toward $\ell^t$, and $b^t$ remains WEQX toward every $i\in N_{r'}\setminus\{b^t\}$.
	
	For any $i\in N_r\setminus\{\ell^t\}$, we have $\frac{v_i(X^{t+1}_i)}{w_i}=\frac{v_i(X^{t}_i)}{w_i}\geq \frac{v_{\ell^t}(X^{t}_{\ell^t})}{w_{\ell^t}}=\frac{v_{\ell^t}(X^{t+1}_{\ell^t}\setminus\{e^t\})}{w_{\ell^t}}=\hat{v}^{t+1}_{\ell^t}$.
	The inequality holds since $\ell^t$ is the least valuer in $N_r$. The last equality holds by Lemma \ref{lem: WEQX price and transferred good}.
	
	For any $i\in N_{r'}\setminus\{b^t\}$, we have $\frac{v_{b^t}(X^{t+1}_{b^{t}})}{w_t}=\frac{v_{b^t}(X^{t}_{b^{t}}\setminus\{e^t\})}{w_t}=\hat{v}^t_{b^t}\geq \hat{v}^{t+1}_{b^{t+1}}\geq \hat{v}^{t+1}_{i}$.
	The second equality is due to Lemma \ref{lem: WEQX price and transferred good}. The first inequality is due to Invariant \ref{inv: WEQX big spending}.
\end{proof}

\begin{lemma}
	The Raised Group Invariant (Invariant \ref{inv: WEQX raised group}) are maintained at the end of round $t$.
\end{lemma}
\begin{proof}
	Clearly, $N_r$ is raised in line \ref{line: WEQX raise price} in the $r$-th iteration of the \textbf{for} loop.
	Thus, at the beginning of the $r^*$-th iteration of the \textbf{for} loop, $N_1,\ldots,N_{r^*-1}$ are raised exactly once, and $N_{r^*},\ldots,N_R$ are not raised.
	\begin{enumerate}
		\item For any $r<r^*$, assume that $N_r$ is raised in round $t$. By Invariant \ref{inv: WEQX least valuer},
		$X^t_j=X^0_j$ for any $j\in N_r$. Thus, only goods in $\cup_{j\in N_{r}} X^0_j$ are raised in this round. The property then follows.
		\item For any $e\in \cup_{j\in N_{< r^*}} X_j^0$, note that $\alpha_{i,e}\leq 1$ in $(X^0,p^0)$ by Lemma \ref{lem: WEQX properties of initial allocation}.
		By property 1, $e$ has been raised. Thus, $\alpha_{i,e}\leq 1/k$ after $e$ is raised.
		\item For any $e\in\cup_{j\in N_{\geq r^*}} X_j^0$, $v_i(e)=1$, $p(e)=k$ and hence $\alpha_{i,e}=1/k$ in $(X^0,p^0)$ by Lemma \ref{lem: WEQX properties of initial allocation}, since $i\in N_{<r^*}$ is in a lower group than any $j\in N_{\geq r^*}$. By property 1, $e$ is not raised. Thus, $\alpha_{i,e}=1/k$ is preserved.
		\item For all $i\in N_{<r^*}$, $\alpha_i=1/k$ directly follows from properties 2 and 3.
		To show $X_i^0\subseteq X_i$, assume that $i\in N_r$ and $N_r$ is raised in round $t'$. By Invariant \ref{inv: WEQX least valuer},
		$X^{t''}_i=X_i^0$ for all $t''\leq t'$.
		After round $t'$, when $i$ serves as a least valuer, she receives new goods. when $i$ serves as a big valuer, she loses a good $e\in X_i\setminus X^0_i$.
		In both cases, $X^0_i\subseteq X_i$ persists.
		\item For all $i\in N_{\geq r^*}$, $\alpha_i=1$ in $(X^0,p^0)$ by Lemma \ref{lem: WEQX properties of initial allocation}. Since $i$ is unraised, $i$ can only serves as a big valuer and loses goods in $X_i^0$. Thus, $X_i\subseteq X^0_i$ and $\alpha_i=1$ is preserved.
	\end{enumerate}
\end{proof}

\begin{lemma}
	The Big Value Invariant (Invariant \ref{inv: WEQX big spending}) are maintained at the end of round $t$.
\end{lemma}
\begin{proof}
	Clearly, the price-rise rounds do not affect the invariant. It suffices to consider the transfer rounds and show that $\hat{v}^{t}_{b^{t}}\geq \hat{v}^{t+1}_{i}$ for any $i\in N$.
	
	For a transfer round $t$, $X^{t+1}_{\ell^{t}}=X^{t}_{\ell^{t}}\cup \{e^t\}, X^{t+1}_{b^{t}}=X^{t}_{b^{t}}\setminus \{e^t\}$, and $X^{t+1}_{i}=X^{t}_{i}$ for all $i\notin\{\ell^t,b^t\}$. Thus, only the value of $\ell^t$ increases and it suffices to consider $\ell^t$. By Lemma \ref{lem: WEQX price and transferred good}, we have $\hat{v}^{t+1}_{\ell^{t}}=\frac{v_{\ell^t}(X^{t+1}_{\ell^{t}} \setminus \{e^t\})}{w_{\ell^{t}}}=\frac{v_{\ell^t}(X^{t}_{\ell^{t}})}{w_{\ell^{t}}} < \hat{v}^{t}_{b^t}$. The inequality is due to the loop condition.
\end{proof}

\begin{lemma}
	The Least Spender Invariant (Invariant \ref{inv: WEQX least valuer}) are maintained at the end of round $t$.
\end{lemma}
\begin{proof}
	Assume that $i=b^{t'}$ at some previous round $t'\leq t$. Then, $\ell^t\in Q^{t'}\subseteq Q^t$ By Lemma \ref{lem: WEQX some unraised groups are in Q} and Invariant \ref{inv: WEQX monotone Q}. Thus, $\frac{v_{\ell^t}^t(X^t_{\ell^t})}{w_{\ell^t}}\geq \hat{v}^t_{b^t}$, which contradicts to the conditions of the \textbf{while} loop or the \textbf{if} statement.
	As a result, $i$ has never been identified as a least valuer (who receives goods) or a big valuer (who loses goods) in a transfer round before. Then, $X^{t'}_i=X^0_i$ for all $t'\leq t$.
\end{proof}

\begin{lemma}
	The $Q$ Invariant (Invariant \ref{inv: WEQX monotone Q}) are maintained at the end of round $t$.
\end{lemma}
\begin{proof}
	Assume we are in the $r$-th iteration of the \textbf{for} loop processing group $N_r$.
	Clearly, the price-rise rounds do not affect the invariant. It suffices to consider the transfer rounds.
	
	For a transfer round $t$, for any $i\in Q^t\setminus\{b^t\}$, we have $\frac{v_i(X^{t+1}_i)}{w_i}\geq \frac{v_i(X^{t}_i)}{w_i}\geq \hat{v}^t_{b^t}\geq \hat{v}^{t+1}_{b^{t+1}}$.
	The last inequality is due to Invariant \ref{inv: WEQX big spending}. Thus, $Q^t\setminus\{b^t\}\subseteq Q^{t+1}$.
	For $b^t$, we have $\frac{v_{b^t}(X^{t+1}_{b^t})}{w_{b^t}}= \frac{v_{b^t}(X^{t}_{b^t}\setminus\{e^t\})}{w_{b^t}}=\hat{v}^t_{b^t}\geq \hat{v}^{t+1}_{b^{t+1}}$.
	The second equality is due to Lemma \ref{lem: WEQX price and transferred good}. The inequality is due to Invariant \ref{inv: WEQX big spending}. Thus, $Q^t\subseteq Q^{t+1}$.
\end{proof}

\section{Conclusion}
\label{sec: conclusion}

In this paper, we have revisited the problem of finding fair and efficient allocations for goods with bivalued utilities. Our main contributions are as follows:
\begin{enumerate}
	\item \textbf{Correcting and Generalizing Prior Work on EFX:} We identified a termination issue in the polynomial-time algorithm proposed by Garg and Murhekar \cite{GargM21} for computing EFX and fPO allocations. To address this, we designed a new polynomial-time algorithm based on the Fisher market that computes a \emph{weighted} EFX (WEFX) and fPO allocation. This not only remedies the flaw in the previous approach but also generalizes the known result for (unweighted) EFX and fPO allocations \cite{BuLLLT24} to the weighted setting.
	\item \textbf{Extending Results to Equitability (EQX):} We adapted our algorithmic framework to compute weighted equitable allocations up to any good (WEQX) that are also fPO. This generalizes the previous polynomial-time result for EQX and PO allocations in the bivalued setting \cite{GargM21} to the weighted case.
\end{enumerate}

Our work advances the state-of-the-art for weighted fair division with bivalued preferences, demonstrating the versatility of Fisher market-based techniques.

\paragraph{Future Directions.} Several open questions arise from our work. First, it is interesting to investigate whether a WEFX and fPO allocation can be computed in polynomial time for the more general model of \emph{personalized} bivalued goods (where the two values may differ across agents). Second, exploring the existence and efficient computation of WEFX and fPO allocations for \emph{bivalued chores} (undesirable items) presents a natural and challenging direction for future research.

\bibliographystyle{plain}
\bibliography{document}

\appendix

\section{A Counterexample}
\label{sec: a counterexample}

For clarity, we present the algorithm from \cite{GargM21} as Algorithm \ref{alg: Garg and Murhekar}, using the notation consistent with this paper.
Algorithm \ref{alg: Garg and Murhekar} seeks an EFX allocation by iteratively adjusting allocations and prices. Its main loop (lines 4-7) transfers goods along MBB paths, serving the same purpose as the loop in our Algorithm \ref{alg: initial allocation}. Subsequently, the algorithm checks a strong termination condition against the group of agents \emph{not} reachable from the least spender $i$ (line 9). If this condition is not met, it raises the prices of all items belonging to agents in $C_i$ (line 12) and repeats. The process continues until the least spender $i$ satisfies the pWEFX condition.
It is this overly strong termination condition that causes the algorithm to potentially run forever.

\begin{algorithm}[tb]
	\caption{EFX and fPO allocation for bivalued goods \cite{GargM21}}
	\label{alg: Garg and Murhekar}
	\begin{algorithmic}[1]
		\REQUIRE An unweighted bivalued instance $(N, M, \{v_i\}_{i \in N})$
		\STATE $(\x,\p)\gets$ welfare-maximizing allocation, where $p(e)=v_i(e)$ for $e\in \x_i$.
		\STATE Construct the MBB graph $G_{\x}=(N,E)$ by adding an MBB edge from $i$ to $j$ if $\MBB_i\cap \x_j\neq \emptyset$.
		\STATE Let $i \in \operatorname{argmin}\limits_{h \in N} p(X_h)$ be the least spender
		\WHILE{there are a path $i\to i_1\to \cdots \to i_s$ in $G_{\x}$ and good $e\in \x_{i_s}$ s.t $p(X_{i_s} \setminus \{e_s\}) > p(X_i)$}
		\STATE Transfer $e_s$ from $i_s$ to $i_{s-1}$ and update allocation $X$
		\STATE Update the least spender $i \in \operatorname{argmin}\limits_{h \in N} p(X_h)$
		\ENDWHILE
		\STATE Let $C_i$ be the set of agents reachable from $i$ in the MBB graph
		\IF{$\forall  h \notin C_i$ and $\forall e \in X_h$, $p(X_h \setminus \{e\}) \leq p(X_i)$} 
		\RETURN $(X,p)$
		\ELSE
		\STATE Raise prices of goods in $C_i$ by a factor of $k$
		\STATE Repeat from Step 3
		\ENDIF
	\end{algorithmic}
\end{algorithm}

{\ThmCounterexample*}

\begin{proof}
	Consider an instance with 2 agents and 5 goods. The agents' valuations are shown in Table \ref{tab: counterexample}, where the parameter $k$ for bivalued utilities is $5$.
	\begin{table}[hp]
		\label{tab: counterexample}
		\centering
		\caption{An instance on which Algorithm \ref{alg: Garg and Murhekar} never terminates.}
		\begin{tabular}{c c c c c c}
			\hline
			& \(e_1\)  & \(e_2\)  & \(e_3\) & \(e_4\) & \(e_5\) \\
			\hline
			$v_1$   & $5$ & $5$ & $5$ & $1$ & $1$ \\ 
			$v_2$   & $1$ & $1$ & $1$ & $1$ & $5$ \\
			\hline
		\end{tabular}
	\end{table}
	
	We now run Algorithm \ref{alg: Garg and Murhekar} on this instance.
	Initially, goods $e_1, e_2, e_3$ are allocated to agent $1$ with $p(e_1)=p(e_2)=p(e_3)=5$, and $e_4,e_5$ are allocated to agent $2$ with $p(e_4)=1, p(e_5)=5$.
	At this point, agent $2$ is identified as the least spender.
	Then, the algorithm step into the \textbf{while} loop and tries to transfer goods from agent $1$ to agent $2$. However, it is easy to see that $\alpha_2=1$ and $\alpha_{2,e}=1/k$ for any $e\in\{e_1,e_2,e_3\}$. It implies that $e_1,e_2,e_3\notin \MBB_2$ and there is no MBB edge from agent $2$ to agent $1$. As a result, the \textbf{while} loop breaks without execution. Next, the algorithm comes to the \textbf{if} statement.
	Since $p(X_2)=6<10=p(X_1\setminus\{e\})$ for any $e\in\{e_1,e_2,e_3\}$, the algorithm will raise the prices of goods in $X_2$ by a multiplicative factor of $5$, resulting in $p(e_1)=p(e_2)=p(e_3)=p(e_4)=5,p(e_5)=25$.
	
	The algorithm then repeats with agent $1$ now being the least spender. In the \textbf{while} loop, it tries to transfer goods from agent $2$ to agent $1$. However, it is easy to see that $\alpha_1=1$ and $\alpha_{1,e}=1/k$ for any $e\in\{e_4,e_5\}$.
	It implies that $e_4,e_5\notin \MBB_1$ and there is no MBB edge from agent $1$ to agent $2$. As a result, the \textbf{while} loop breaks without execution. For the \textbf{if} statement, since $p(X_1)=15<25=p(e_5)=p(X_2\setminus\{e_4\})$, the algorithm will raise the prices of goods in $X_1$ by a multiplicative factor of $5$, resulting in $p(e_1)=p(e_2)=p(e_3)=25, p(e_4)=5,p(e_5)=25$. At the moment, agent $2$ again becomes the least spender, and everything remains the same as that at the beginning except that the price of all goods are raised by a factor of $5$.
	
	To conclude, the algorithm enters an infinite loop. The allocation $X$ remains unchanged indefinitely, while the prices of all goods are multiplied by a factor of $5$ every two iterations (completing one full cycle: agent $2\to$ agent $1\to$ agent $2$ as least spender). Therefore, Algorithm \ref{alg: Garg and Murhekar} does not terminate on this instance.
\end{proof}

\end{document}